\documentclass[aps,prx,twocolumn,nopacs,superscriptaddress,groupedaddress,longbibliography]{revtex4-1}  %
\usepackage{graphicx}  %
\usepackage{dcolumn}   %
\usepackage{bm}        %
\usepackage{amssymb}   %
\usepackage{amsmath}
\usepackage{braket}
\usepackage{amsfonts}
\usepackage{hyperref}
\usepackage{mathtools}
\usepackage{xcolor}
\usepackage{multirow}

\DeclarePairedDelimiter\floor{\lfloor}{\rfloor}
\usepackage{amsthm}
\newtheorem{theorem}{Theorem}

\usepackage{mathtools}

\begin{document}

\title{Topological correlations in three dimensional classical Ising models: an exact solution with a continuous phase transition}
\author{Zhiyuan Wang}
\affiliation{Department of Physics and Astronomy, Rice University, Houston, Texas 77005,
  USA}
\affiliation{Rice Center for Quantum Materials, Rice University, Houston, Texas 77005, USA}
\author{Kaden R.~A. Hazzard}
\affiliation{Department of Physics and Astronomy, Rice University, Houston, Texas 77005,
USA}
\affiliation{Rice Center for Quantum Materials, Rice University, Houston, Texas 77005, USA}
\date{\today}

\begin{abstract}
We study a three-dimensional~(3D) classical Ising model that is exactly solvable when some coupling constants take certain imaginary values.
The solution combines and generalizes the Onsager-Kaufman solution~\cite{Onsager1944,Kaufman1949} of the 2D Ising model and the solution of Kitaev's honeycomb model~\cite{Kitaev2006}, leading to a three-parameter phase diagram with a third order phase transition between two distinct phases. 
Interestingly, the phases of this model are distinguished by topological features: the expectation value of a certain family of loop observables depend only on the topology of the loop~(whether the loop is contractible), and are quantized at rational values that differ in the two phases. 
We show that a related exactly solvable 3D classical statistical model with real coupling constants also shows the topological features of one of these phases. 
Furthermore, even in the model with complex parameters, the partition function has some physical relevance, as it can be interpreted as the transition amplitude of a quantum dynamical process and may shed light on dynamical quantum phase transitions.  %
\end{abstract}
\maketitle

\section{Introduction}
Understanding the universal behavior of classical many-body systems near their critical points is a central goal %
of classical statistical mechanics. Although this is a difficult problem in general, in one and two spatial dimensions, significant insights have been provided by exactly solved models~\cite{baxter2016exactly}. One important open problem is to generalize these solutions to three-dimensional~(3D) systems with realistic short-range interactions. Despite a long effort with some preliminary results~\cite{suzuki1972solution,zamolodchikov1980tetrahedra,bazhanov1992new,huang1997exact,dhar2008exact,mangazeev2013integrable}, no physical 3D model has been exactly solved that displays a genuinely 3D phase transition~\footnote{Among the models constructed in Refs.~\cite{suzuki1972solution,zamolodchikov1980tetrahedra,bazhanov1992new,huang1997exact,dhar2008exact,mangazeev2013integrable}, only the models in  Refs.~\cite{suzuki1972solution,huang1997exact} have phase transitions, and in these the 3D partition function factorizes into a product of partition functions of 2D systems, giving the phase transitions an essentially 2D character.}.

In this paper, we make progress in this direction by exactly solving a classical Ising model on a special 3D lattice, as depicted in Fig.~\ref{fig:hc0}, 
although with the caveat that the model has imaginary coupling constants. 
The transfer matrix of this system has a  structure  similar to  a non-Hermitian version of the 2D Kitaev honeycomb model~\cite{Kitaev2006}, and the partition function can be obtained using the representation theory of the so($2N$) Lie algebra and the corresponding Lie group.
The solution displays a third order phase transition between two distinct phases, and near the critical point we can exactly obtain a critical exponent of the model. 

The phases are interesting in their own right, as they are distinguished by topological properties. Specifically, there is a family of loop observables whose 
 expectation values distinguish the two phases and are equal to some rational numbers~($0$, $1$, or $1/3$) depending on the topology of the loop. 

Despite its complex coupling constants~(also a complication of some previous approaches~\cite{zamolodchikov1980tetrahedra,bazhanov1992new}), our findings have physical relevance. First, we show in Sec.~\ref{sec:physical_model} that the topological features discovered in one of the phases of the model with complex couplings also exist in a similar exactly solvable model with real-valued couplings. More speculatively, it is possible more generally that the long-distance property of our model belongs to the same universality class of certain physical 3D classical systems. 
It remains an open question whether the other phase of our model can also be  reproduced in a physical system, but if there indeed exists a physical classical system that has the two phases mentioned above and a phase transition between them, then the concept of universality suggests that the long-distance behaviors and the critical exponent we obtain here will apply to
such physical systems.

As another point of physical relevance for the model with complex couplings, in Sec.~\ref{sec:quantum_amplitude} we show two constructions that realize the partition function $Z$ of our model in certain dynamical processes of a 3D quantum system: one is to map $Z$ to the transition amplitude between a family of product states, the other is to realize $Z$ as the coherence of a probe spin coupled to the whole 3D system. Both constructions in principle allow the free energy to be experimentally measured, albeit with an exponentially small signal. Under these mappings, the phase transition of our model corresponds to a dynamical quantum phase transition~(DQPT)~\cite{heyl2013dynamical,Heyl_2018}, a phenomenon that has gained much attention recently. Statistical mechanics with complex configuration energies also appears in  the study of Lee-Yang zeros~\cite{yang1952statistical,lee1952statistical,wei2012lee,peng2015experimental}, non-Hermitian quantum systems~\cite{moiseyev2011non,Gong2018Topological,Ashida2020Non}, and complex conformal field theories~\cite{faedo2020holographic}.   

Our paper is organized as follows. In Sec.~\ref{sec:model} we define our model and a family of loop observables of interest. In Sec.~\ref{sec:solution} we present the exact solution of  the model: in Sec.~\ref{sec:TM} we derive the transfer matrix of the classical model, in Sec.~\ref{sec:map_fermion} we use a spin-fermion mapping to reduce the problem to a free fermion problem, in Sec.~\ref{sec:solve_fermion_TM} we solve the eigenvalues of the free fermion transfer matrix and calculate the thermodynamic free energy, in Sec.~\ref{sec:phase_boundary} we obtain the phase diagram, in Sec.~\ref{sec:critical_exp} we calculate a critical exponent, and in Sec.~\ref{sec:TPloop} and Sec.~\ref{sec:loop_observables} we calculate the expectation values of loop observables and demonstrate their topological properties. In Sec.~\ref{sec:justification} we give two physical implications of our model: the existence of a physical classical phase with similar topological behaviors~(Sec.~\ref{sec:physical_model}), and realizations of the partition function in quantum dynamical processes~(Sec.~\ref{sec:quantum_amplitude}). %
In Sec.~\ref{sec:summary} we summarize our results. 
The Appendices contain technical results used throughout our arguments. 

\section{The Model}\label{sec:model}
\begin{figure}
	\center{\includegraphics[width=0.9\linewidth]{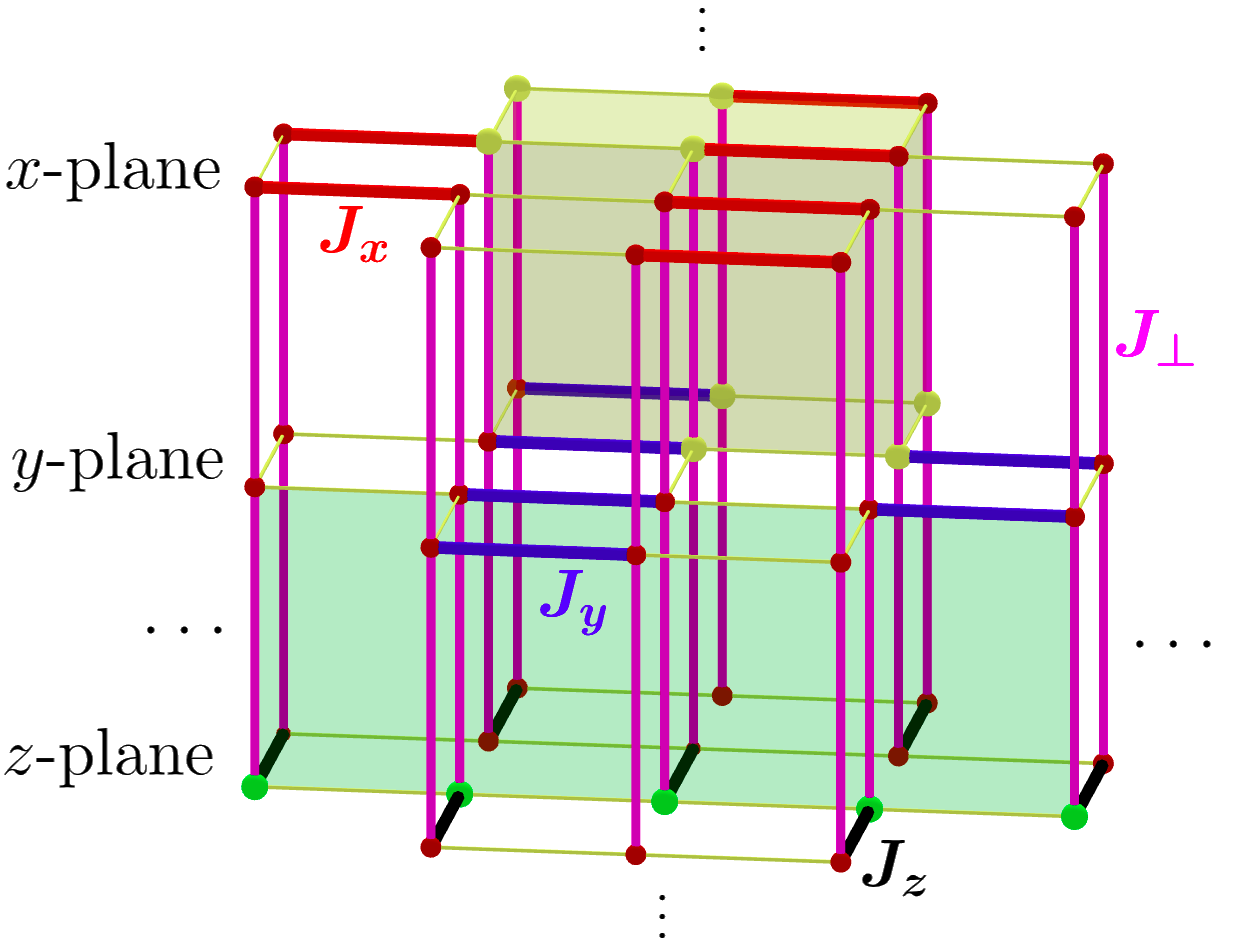}}%
	\caption{\label{fig:hc0} Definition of the model and loop observables. The classical system sits on a 3D stacking of the brick wall lattice, of arbitrarily large extent in each direction. Classical spins lie on vertices, and they only interact via the thicker links. The horizontal links~(red, blue, black) have real coupling constants $J_x,J_y,J_z$, for $x$-planes, $y$-planes, and $z$-planes, respectively. The coupling constant $J_\perp$ for the vertical links~(pink) and the  external field $h$ are imaginary when the solvability condition Eq.~\eqref{cond:exact} is met. The yellow shaded cuboid shows an example of the loop observable $\sigma[\mathfrak{L}_{(xy)}]$ for a contractible loop $\mathfrak{L}$~(here being an elementary plaquette), which is equal to the product of Ising spins on the larger yellow vertices~[see Eq.~\eqref{eq:def_loop_observable}]. Similarly, the green shaded rectangle shows an example of $\sigma[\mathfrak{L}_{(yz)}]$ for a noncontractible loop, extended infinitely to the right and to the left.
	}
\end{figure}
In this section we define our model and the class of physical observables we are interested in.
The model is defined on a 3D stacking of the 2D brick wall lattice, with classical Ising spins, $\sigma_j \in \{-1, +1\}$, lying on vertices $j$, as shown in Fig.~\ref{fig:hc0}, and we use periodic boundary conditions~(PBC) for all the three directions for simplicity. Nearest neighbor Ising-type interactions exist only on a subset of links in this lattice, which are shown in Fig.~\ref{fig:hc0} as thick red, blue, black, and pink links. The  energy of the system for a specific classical spin configuration is
\begin{eqnarray}\label{eq:HIK}
  H[\{\sigma\}]&=&-J_x\sum_{\langle ij\rangle\in \mathbf{X}}\sigma_{i}\sigma_{j}-J_y\sum_{\langle ij\rangle\in \mathbf{Y}}\sigma_{i}\sigma_{j}\\
  &&-J_z\sum_{\langle ij\rangle\in \mathbf{Z}}\sigma_{i}\sigma_{j}-J_\perp\sum_{\langle ij\rangle\in\boldsymbol{\perp}}\sigma_{i}\sigma_{j}+h\sum_i\sigma_i,\nonumber
 \end{eqnarray}
where $\mathbf{X}$ denotes the set of all thick links on $x$-planes, and similarly for $\mathbf{Y},\mathbf{Z}$, while $\boldsymbol{\perp}$ is the set of all the vertical links in Fig.~\ref{fig:hc0}, and the external field $h$ acts on all spins. 
The goal is to find the partition function
\begin{eqnarray}\label{eq:Z}
Z(K_x,K_y,K_z,K_\perp,\beta h)=\sum_{\{\sigma\}}e^{-\beta H[\{\sigma\}]},
\end{eqnarray}
where $K_i=\beta J_i,i=x,y,z,\perp$. The free energy is related to the partition function by
\begin{equation}\label{def:free_energy0}
	F=-k_B T \ln Z.
\end{equation}
The model is exactly solvable when the following conditions hold:
\begin{equation}\label{cond:exact}
4J_\perp\beta\equiv \pi i~(\mathrm{mod}~2\pi i),~~~2h\beta\equiv \frac{\pi i}{2}~(\mathrm{mod}~ 2\pi i).
\end{equation}
After imposing these solvability conditions, there remains a three-dimensionless-parameter space~$(K_x,K_y,K_z)$ of solutions.

Beyond the free energy~(and its derivatives), we also consider the thermal expectation values of a family of loop observables that are products of $\sigma_j$s on closed loops, defined by the following procedure:\\
(1) Choose a loop $\mathfrak{L}$ on the 2D brick wall lattice~($\mathfrak{L}$ must consist of edges of the brick wall lattice);\\
(2) Choose two nearest neighbor planes of type $\alpha$ and $\beta$ of the 3D lattice, denoted $(\alpha \beta)$, which can be $(xy),(yz)$ or $(zx)$;\\
(3) Denote by $\mathfrak{L}_{(\alpha\beta)}$ the graph consisting of all sites in the loop $\mathfrak{L}$ of both $\alpha$ and $\beta$ planes and the edges of the lattice joining pairs of these sites;\\ 
(4) For a lattice site $i\in \mathfrak{L}_{(\alpha\beta)}$, denote by $\bar{i}$ the same site of the other plane~(if $i\in \alpha$, then $\bar{i}\in \beta$ and vice versa);\\
(5) For  $i\in \mathfrak{L}_{(\alpha\beta)}$, define $n(i)$ to be the number of thick horizontal edges in $\mathfrak{L}_{(\alpha\beta)}$ linked to $i$~[notice that $n(i)\in\{0,1\}$];\\
(6) The loop product is defined as 
\begin{equation}\label{eq:def_loop_observable}
	\sigma[\mathfrak{L}_{(\alpha\beta)}]=\prod_{i\in \mathfrak{L}_{(\alpha\beta)}} \sigma_i^{n(\bar{i})}.
\end{equation}
In Fig.~\ref{fig:hc0} we illustrate the definition of $\sigma[\mathfrak{L}_{(\alpha\beta)}]$ for a contractible and a noncontractible loop $\mathfrak{L}$.
In Sec.~\ref{sec:TPloop} we will compute their thermal expectation values
\begin{eqnarray}\label{eq:thermal_loop}
\langle\sigma[\mathfrak{L}_{(\alpha\beta)}]\rangle=\frac{1}{Z}\sum_{\{\sigma\}}\sigma[\mathfrak{L}_{(\alpha\beta)}] e^{-\beta H[\{\sigma\}]}.
\end{eqnarray}
We will see  that the expectation values of these observables are sensitive to the topology of the loop $\mathfrak{L}_{(\alpha\beta)}$. Namely, for a contractible loop $\mathfrak{L}_{(\alpha\beta)}$ we have $\langle\sigma[\mathfrak{L}_{(\alpha\beta)}]\rangle=\pm 1$~(and the same loop $\mathfrak{L}_{(\alpha\beta)}$ takes the same value for different phases), while for a non-contractible loop $\mathfrak{L}_{(\alpha\beta)}$, $\langle\sigma[\mathfrak{L}_{(\alpha\beta)}]\rangle$ is equal to $0$ in one phase~(the $A$-phase) and $-1/3$ in another phase~(the $B$-phase). Therefore, noncontractible loop observables can be used as order parameters of this model.

\section{The solution}\label{sec:solution}
\begin{figure}
	\center{\includegraphics[width=0.75\linewidth]{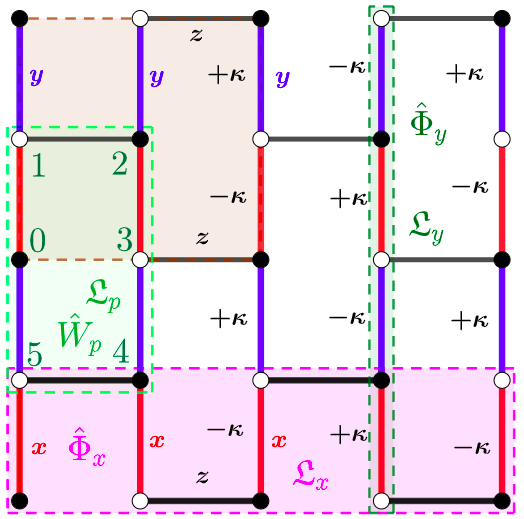}}
	\caption{\label{fig:brickwall-unitcell} The 2D brick wall lattice on which the transfer matrix Eq.~\eqref{eq:def_TM} is defined. A unit cell is shown in the shaded square. The conserved loop operator $\hat{W}_p$ acts on the six spins of the elementary plaquette $p$, and the conserved noncontractible loop $\hat{\Phi}_x$~($\hat{\Phi}_y$) acts on a row~(column) of spins. The $\pm \kappa$ shown next to each link is the real part of the small perturbation to the link coupling constant needed to gap the fermionic spectrum of the $B$-phase. }
\end{figure}  
\subsection{The Transfer Matrix}\label{sec:TM}
The first step  to solve this model is to find the transfer matrix $\hat{T}$ for each period of $x,y,z$ planes, as shown in Fig.~\ref{fig:hc0}, defined so that $Z=\mathrm{Tr}[\hat{T}^M]$, where $M$ is the total number of periods. We will show that when the conditions~\eqref{cond:exact} are satisfied, the transfer matrix is
\begin{eqnarray}\label{eq:def_TM}
  \hat{T}&=&\exp\left(K_x\sum_{\langle ij\rangle\in X_{2D}}\hat{\sigma}^x_{i}\hat{\sigma}^x_{j}\right)\exp\left(K_y\sum_{\langle ij\rangle\in Y_{2D}}\hat{\sigma}^y_{i}\hat{\sigma}^y_{j}\right)\nonumber\\
  &&\times\exp\left(K_z\sum_{\langle ij\rangle\in Z_{2D}}\hat{\sigma}^z_{i}\hat{\sigma}^z_{j}\right),
\end{eqnarray}
where $\hat{\sigma}^{x,y,z}_i$ are Pauli operators acting on the spin located at site $j$ of the 2D brick wall lattice shown in Fig.~\ref{fig:brickwall-unitcell}, $X_{2D}$ denotes the set of all the $x$-links shown in Fig.~\ref{fig:brickwall-unitcell}, and similarly for $Y_{2D},Z_{2D}$. Henceforth, we use $\sum_{x}$, $\sum_{y}$, and $ \sum_{z}$ as abbreviations for  $\sum_{(i,j) \in  X_{2D}}$, $\sum_{(i,j) \in  Y_{2D}}$, and $\sum_{(i,j) \in  Z_{2D}}$, respectively.
We prove Eq.~\eqref{eq:def_TM} by inserting resolutions of identity on each plane in $Z=\mathrm{Tr}[\hat{T}^M]$ in Eq.~\eqref{eq:def_TM} and showing that it reproduces
Eq.~\eqref{eq:Z}. The trick here is that when inserting resolution of identity, we use the $\hat{\sigma}^x$ basis $|X\rangle\equiv\otimes_j|\sigma_j\rangle_x$ on $x$-planes, $\hat{\sigma}^y$ basis $|Y\rangle\equiv\otimes_j|\sigma_j\rangle_y$ on $y$-planes, and $\hat{\sigma}^z$ basis $|Z\rangle\equiv\otimes_j|\sigma_j\rangle_z$ on $z$-planes, where $\hat{\sigma}^x_j|\sigma_j\rangle_x=\sigma_j|\sigma_j\rangle_x$ and similarly for $|\sigma_j\rangle_y,|\sigma_j\rangle_z$. Therefore, we have
\begin{widetext}
\begin{eqnarray}\label{eq:quantum_classical_map}
  \mathrm{Tr}[\hat{T}^M]&=&\sum_{\substack{X_1,Y_1,Z_1,\ldots,\\ X_M,Y_M,Z_M}}\langle X_1|e^{K_x\sum_{x}\hat{\sigma}^x_i\hat{\sigma}^x_j}|Y_1\rangle \langle Y_1|e^{K_y\sum_{y}\hat{\sigma}^y_i\hat{\sigma}^y_j}|Z_1\rangle \langle Z_1|e^{K_z\sum_{z}\hat{\sigma}^z_i\hat{\sigma}^z_j}|X_2\rangle\langle X_2|\cdots\nonumber\\
                        &&~~~~~~~~{}\times|X_M\rangle\langle X_M|e^{K_x\sum_{x}\hat{\sigma}^x_i\hat{\sigma}^x_j}|Y_M\rangle \langle Y_M|e^{K_y\sum_{y}\hat{\sigma}^y_i\hat{\sigma}^y_j}|Z_M\rangle \langle Z_M|e^{K_z\sum_{z}\hat{\sigma}^z_i\hat{\sigma}^z_j}|X_1\rangle\nonumber\\
                        &=&\sum_{\{\sigma\}}\exp\left(K_x\sum_{\langle ij\rangle\in \mathbf{X}} \sigma_i \sigma_j+K_y\sum_{\langle ij\rangle\in \mathbf{Y}} \sigma_i \sigma_j+K_z\sum_{\langle ij\rangle\in \mathbf{Z}} \sigma_i \sigma_j\right)\nonumber\\
                        &&~~~~{}\times\langle X_1|Y_1\rangle \langle Y_1|Z_1\rangle \langle Z_1|X_2\rangle\cdots
  \langle X_M|Y_M\rangle \langle Y_M|Z_M\rangle \langle Z_M|X_1\rangle.
\end{eqnarray}
\end{widetext}
The first factor corresponds to the classical Boltzmann weight contributed by all the horizontal links. For the overlap matrices in the last line of Eq.~\eqref{eq:quantum_classical_map}, using a suitable phase convention for basis states 
\begin{eqnarray}
	|\pm 1\rangle_z&=&\{|\uparrow\rangle,e^{\pi i/4}|\downarrow\rangle \},\nonumber\\
	|\pm 1\rangle_x&=&\{\frac{|\uparrow\rangle+|\downarrow\rangle}{\sqrt{2}}e^{3\pi i/4},\frac{|\uparrow\rangle-|\downarrow\rangle}{\sqrt{2}}e^{\pi i/2} \}\nonumber\\
	|\pm 1\rangle_y&=&\{\frac{|\uparrow\rangle+i|\downarrow\rangle}{\sqrt{2}}e^{-3\pi i/4},\frac{|\uparrow\rangle-i|\downarrow\rangle}{\sqrt{2}}e^{-\pi i/2} \},\nonumber
\end{eqnarray}
we have $${}_x\langle \sigma|\sigma'\rangle_y={}_y\langle \sigma|\sigma'\rangle_z ={}_z\langle \sigma|\sigma'\rangle_x  =\frac{1}{\sqrt{2}}e^{\frac{\pi i}{4}(\sigma\sigma'-\frac{\sigma+\sigma'}{2}+3)}.$$ The overlaps give the Boltzmann weights contributed by the vertical links and external field terms with $\beta J_\perp=\pi i/4,\beta h=\pi i/4$, up to an irrelevant constant shift of the energy. Also, one can show that adding $2\pi i $ to $4J_\perp\beta$ or $2h\beta$ will only multiply the partition function by an irrelevant overall constant phase factor, since whenever we flip a spin $\sigma_j$, the imaginary part of $\beta H[\{\sigma\}]$ changes by $\pm 2( 2s J_\perp-  h )\beta $, where $s=(\sigma'_j+\sigma''_j)/2\in\{-1,0,+1\}$, and $\sigma'_j$~($\sigma''_j$) is the neighbor of $\sigma_j$ lying above~(below) it. Therefore the model Eq.~\eqref{eq:HIK} has transfer matrix Eq.~\eqref{eq:def_TM} when Eq.~\eqref{cond:exact} is satisfied.

Now that we have obtained the transfer matrix $\hat{T}$ of our model, the next step is to calculate the largest~(in magnitude) eigenvalue $\Lambda_{\max}$ of $\hat{T}$, which governs the free energy in the thermodynamic limit 
\begin{equation}\label{def:free_energy}
	F =-M k_B T \ln\Lambda_{\max}+O(\Lambda_1^M/\Lambda_{\max}^M),
\end{equation}
where $\Lambda_1$ is the next-to-largest~(in magnitude) eigenvalue of $\hat{T}$. We will calculate the eigenvalues of $\hat{T}$ in two steps: in Sec.~\ref{sec:map_fermion} we map the transfer matrix $\hat{T}$ to a free fermion transfer matrix $\hat{T}'$ in Eq.~\eqref{eq:Tbondfermion}, and then in Sec.~\ref{sec:solve_fermion_TM} we solve the eigenvalues of this free fermion transfer matrix.

\subsection{Mapping to a free fermion problem}\label{sec:map_fermion}
Our goal in this section is to map the transfer matrix $\hat{T}$ to a free fermion transfer matrix $\hat{T}'$, written in terms of Majorana fermion bilinear operators. While this can be accomplished by Kitaev's original technique~\cite{Kitaev2006}, or by using a Jordan-Wigner transformation~\cite{Feng2007JordanWigner},
here we use the algebraic method developed in Refs.~\cite{Nussinov2009bond,Cobanera2011bond,Chapman2020characterizationof,Ogura2020geometric}, which is far simpler. The key idea of this technique is that, instead of considering the mapping of each individual spin operators, we view the interaction term on each link $\langle ij\rangle$ as a whole, and consider the algebra generated by all these terms. We write the transfer matrix as 
\begin{equation}\label{eq:Tbond}
	\hat{T}=e^{K_x\sum_{x}\hat{\gamma}_{ij}}e^{K_y\sum_{y}\hat{\gamma}_{ij}}e^{K_z\sum_{z}\hat{\gamma}_{ij}},
\end{equation}
where the bond operators are defined as $\hat{\gamma}_{ij}=\hat{\sigma}_i^\alpha\hat{\sigma}_j^\alpha$ if $\langle ij\rangle$ is an $\alpha$-link in the 2D brick wall lattice. We now construct another transfer matrix
\begin{eqnarray}\label{eq:Tbondfermion}
	\hat{T}'&=&e^{K_x\sum_{x}\hat{\gamma}'_{ij}}e^{K_y\sum_{y}\hat{\gamma}'_{ij}}e^{K_z\sum_{z}\hat{\gamma}'_{ij}}\nonumber\\
	&\equiv& e^{K_x\sum_{x}u_{ij}i\hat{c}_i\hat{c}_j}e^{K_y\sum_{y}u_{ij}i\hat{c}_i\hat{c}_j}e^{K_z\sum_{z}u_{ij}i\hat{c}_i\hat{c}_j},
\end{eqnarray}
which has exactly the same exponential structure  and the same set of parameters as $\hat{T}$, but has the bond operators replaced by Majorana fermion bilinears $\hat{\gamma}'_{ij}\equiv u_{ij}i\hat{c}_i\hat{c}_j$ on each link, where $\hat{c}^\dagger_i=\hat{c}_i$, and $\{\hat{c}_i,\hat{c}_j\}=2\delta_{ij}$.
Here $u_{ij}$ is a real number defined independently on each link,  whose value is to be determined later. Notice that the ordering of Majorana operators $\hat{c}_i\hat{c}_j$ matters in the sum since they anti-commute; throughout this paper, we use the convention that whenever we sum~(or product) over links, each link $\langle ij\rangle$ appears only once in the sum, with $i$ representing an even site~(black dots in Fig.~\ref{fig:brickwall-unitcell}) and $j$ representing an odd site ~(white open circles in Fig.~\ref{fig:brickwall-unitcell}), and we always order $\hat{c}_i$ to the left unless otherwise stated. 

The goal now is to choose these real coefficients $\{u_{ij}\}$ \textit{such that the algebra generated by $\{\hat{\gamma}_{ij}\}$ is isomorphic to the algebra generated by $\{\hat{\gamma}'_{ij}\}$.} Once this is done, Refs.~\cite{Nussinov2009bond,Cobanera2011bond,Chapman2020characterizationof,Ogura2020geometric} claim that there exists a unitary mapping $\hat{U}$ between the two systems such that $\hat{\gamma}'_{ij}=\hat{U} \hat{\gamma}_{ij}\hat{U}^\dagger$ for all links $\langle ij\rangle$~(we will also need to check that the Hilbert space dimensions of the two systems are the same), leading to $\hat{T}'=\hat{U}\hat{T}\hat{U}^\dagger$, i.e. $\hat{T}$ and $\hat{T}'$ have the same eigenvalues.  Requiring the two algebras to be isomorphic  means that any algebraic relation satisfied by the generators $\{\hat{\gamma}_{ij}\}$, say $f(\{\hat{\gamma}_{ij}\})=0$, must be satisfied by $\{\hat{\gamma}'_{ij}\}$ as well, $f(\{\hat{\gamma}'_{ij}\})=0$, and vice versa. In our case, this leads to the following four families of relations:\\
\textit{Relation 1.} We have $\hat{\gamma}^2_{ij}=1$ for each link $\langle ij\rangle$, and therefore we must require $\hat{\gamma}'^2_{ij}=u_{ij}^2=1$, which constrains $u_{ij}$ to be $\pm 1$.\\
\textit{Relation 2.} Two bond operators anti-commute if and only if they share exactly one vertex, otherwise, they commute. It is straightforward to check that this is satisfied by both $\{\hat{\gamma}_{ij}\}$ and $\{\hat{\gamma}'_{ij}\}$, so this condition puts no  constraints on $\{u_{ij}\}$.\\
\textit{Relation 3.} The product of $\hat{\gamma}'_{ij}$ on any closed loop $\mathfrak{L}$ is equal to a constant, so the  product of $\hat{\gamma}_{ij}$ on $\mathfrak{L}$ must be equal to the same constant. It is enough to require this constraint only on all the elementary plaquettes $\mathfrak{L}_p$ along with two large loops $\mathfrak{L}_x$ and $\mathfrak{L}_y$ winding around the torus~(as shown in Fig.~\ref{fig:brickwall-unitcell}),
since the product on other loops decompose into products on these elementary loops. The product of $\hat{\gamma}'_{ij}$ on these loops are equal to 
\begin{eqnarray}\label{eq:Wpphixphiy}
	W_p&\equiv& \prod_{\langle ij\rangle\in \mathfrak{L}_p} \hat{\gamma}'_{ij}=\prod_{\langle ij\rangle\in \mathfrak{L}_p} u_{ij},\nonumber\\%\hat{\gamma}'_{01}\hat{\gamma}'_{21}\hat{\gamma}'_{23}\hat{\gamma}'_{43}\hat{\gamma}'_{45}\hat{\gamma}'_{05}
	\Phi_x&\equiv & \prod_{\langle ij\rangle\in \mathfrak{L}_x} \hat{\gamma}'_{ij}=\prod_{\langle ij\rangle\in \mathfrak{L}_x} u_{ij},\nonumber\\
	\Phi_y&\equiv& \prod_{\langle ij\rangle\in \mathfrak{L}_y} \hat{\gamma}'_{ij}=\prod_{\langle ij\rangle\in \mathfrak{L}_y} u_{ij},
\end{eqnarray}
for every plaquette $p$, and we order the product of operators according to their linear order in the loop~(the orientation of the loop and the initial point do not affect the result of the product).

The product of $\hat{\gamma}_{ij}$ on these loops are equal to 
\begin{eqnarray}\label{eq:phixphiy}
	\hat{W}_p&\equiv& \prod_{\langle ij\rangle\in \mathfrak{L}_p} \hat{\gamma}_{ij}=-\hat{\sigma}^z_0\hat{\sigma}^y_1\hat{\sigma}^y_2\hat{\sigma}^z_3\hat{\sigma}^x_4\hat{\sigma}^x_5,\nonumber\\
	\hat{\Phi}_x&\equiv & \prod_{\langle ij\rangle\in \mathfrak{L}_x} \hat{\gamma}_{ij}=-\prod_{i\in \mathfrak{L}_x} \hat{\sigma}^y_i,\nonumber\\
	\hat{\Phi}_y&\equiv& \prod_{\langle ij\rangle\in \mathfrak{L}_y} \hat{\gamma}_{ij}=-\prod_{i\in \mathfrak{L}_y} \hat{\sigma}^z_i,
\end{eqnarray}
where $0,1,2,3,4,5$  label the sites of the plaquette $p$, as shown in Fig.~\ref{fig:brickwall-unitcell}~(and similarly for all other plaquettes).
Although the RHS of Eq.~\eqref{eq:phixphiy} are not constants, one can check that these operators mutually commute, and they commute with all the bond operators $\hat{\gamma}_{ij}$, and therefore they commute with the transfer matrix $\hat{T}$. They play the role of conserved observables, and their common eigenspaces are invariant under the action of $\hat{T}$. Further, since $\hat{W}_p^2=\hat{\Phi}_x^2=\hat{\Phi}_y^2=1$, their eigenvalues can only be $\pm 1$. 
To guarantee the algebraic isomorphism between the algebras $\{\hat{\gamma}_{ij}\}$ and $\{\hat{\gamma}'_{ij}\}$, we need to 
map the spin model transfer matrix $\hat{T}$ in each common eigenspace of $\{\hat{W}_p,\hat{\Phi}_x,\hat{\Phi}_y\}$ to a different fermionic transfer matrix  $\hat{T}'$, with the $u_{ij}$ chosen in such a way that their loop products $\{W_p,\Phi_x,\Phi_y\}$ equal the eigenvalues of $\{\hat{W}_p,\hat{\Phi}_x,\hat{\Phi}_y\}$.\\
\textit{Relation 4.}
On a closed manifold, the product of all $\{\hat{\gamma}_{ij}\}$ on the lattice equals a constant:
\begin{eqnarray}\label{eq:prod_all_bond}
\prod_{\text{all }\langle ij\rangle} \hat{\gamma}_{ij}=i^{4L_xL_y}=1,
\end{eqnarray}
where $L_x$~($L_y$) is the system size in the $x$-~($y$-) direction.
Similarly, the product of all $\{\hat{\gamma}'_{ij}\}$ is
\begin{eqnarray}\label{eq:prod_all_bond_f}
\prod_{\text{all }\langle ij\rangle} \hat{\gamma}'_{ij}=
\hat{P}_f\prod_{\text{all }\langle ij\rangle}  u_{ij},
\end{eqnarray}
where
$\hat{P}_f\equiv \prod_{z}\left(-i\hat{c}_i\hat{c}_j\right)$
is the conserved fermion parity operator. Therefore the algebraic isomorphism restricts the fermion model to the eigen-subspace of $\hat{P}_f$ with eigenvalue
\begin{equation}\label{eq:parity_restriction}
	P_f= \prod_{\text{all }\langle ij\rangle} u_{ij}.
\end{equation}
\textit{Summary and consistency check.} In summary, the mutually commuting conserved operators $\{\hat{W}_p,\hat{\Phi}_x,\hat{\Phi}_y\}$ split the full Hilbert space into a direct sum of their common eigen-subspaces, and the transfer matrix $\hat{T}$ leaves each subspace invariant. In the subspace labeled by the conserved eigenvalues $\{W_p,\Phi_x,\Phi_y\}$, the transfer matrix $\hat{T}$ is mapped to a fermionic transfer matrix $\hat{T}'$ defined in Eq.~\eqref{eq:Tbondfermion} where the parameters $u_{ij}=\pm 1$ are chosen to satisfy Eq.~\eqref{eq:Wpphixphiy}~\footnote{While there are exponentially many solutions $\{u_{ij}\}$ to  Eq.~\eqref{eq:Wpphixphiy} for a fixed configuration $\{W_p,\Phi_x,\Phi_y\}$, all of them are equivalent up to a gauge transformation, and the spectrum of $\hat{T}'$ only depends on the values of $\{W_p,\Phi_x,\Phi_y\}$.}, and $\hat{T}'$ is restricted to a fixed fermion parity sector satisfying Eq.~\eqref{eq:parity_restriction}. 

As a consistency check, let us verify that the subspace dimension of the spin and fermionic systems, mapped to each other by the above algebraic isomorphism, are the same. For the spin system, we have $4L_xL_y$ qubit degrees of freedom~(d.o.f.) in total; in each subspace, the constraint Eq.~\eqref{eq:phixphiy} removes $2L_x L_y-1+2$ qubit d.o.f~($-1$ because the product of all $\hat{W}_p$ is a constant, so only $2L_x L_y-1$ of them are independent), leaving us with $2L_x L_y-1$ qubit d.o.f. For the fermionic system, we have $4L_xL_y$ Majorana fermions in total, which amounts to $2L_xL_y$ Dirac fermion d.o.f.; the fermion parity restriction Eq.~\eqref{eq:parity_restriction} further removes one of them, leaving us $2L_x L_y-1$ Dirac fermion d.o.f.. Therefore the Hilbert space dimension of the two systems are the same, both equal to $2^{2L_x L_y-1}$.

\subsection{Solving the free fermion transfer matrix}\label{sec:solve_fermion_TM}
In the last section we mapped the transfer matrix $\hat{T}$ in each sector labeled by $\{W_p,\Phi_x,\Phi_y\}$ to a free fermion transfer matrix $\hat{T}'$ in Eq.~\eqref{eq:Tbondfermion}, where $u_{ij}=\pm 1$ are chosen to satisfy Eq.~\eqref{eq:Wpphixphiy}, and the fermion parity satisfies Eq.~\eqref{eq:parity_restriction}. Now we solve these free fermion problems in each sector to get the full spectrum of $\hat{T}$. The difficulty here is that there are exponentially many such sectors~($2^{2L_x L_y+1}$ in total), most of which are not translationally invariant and can only be solved numerically. Fortunately we are most interested in the sector that contains the principal eigenvalue $\Lambda_\mathrm{max}$ of $\hat{T}$, i.e. the sector $\{W_p,\Phi_x,\Phi_y\}$ where the principal eigenvalue of $\hat{T}'$ is largest, since $\Lambda_\mathrm{max}$~(and the corresponding principal eigenstate $|\Lambda_\mathrm{max}\rangle$) determines the thermodynamic properties of the original classical system. 
In App.~\ref{appen:Lieb} we prove a generalization of Lieb's optimal flux theorem~\cite{lieb1994flux} for the transfer matrix $\hat{T}'$, which shows that for real $K_x,K_y,K_z$, the principal eigenvalue of $\hat{T}'$ is maximized by a configuration $\{W_p,\Phi_x,\Phi_y\}$ where all $W_p$ are equal to $+1$. From now on we will call such a configuration  vortex-free, and for a configuration with some $W_p=-1$ we say it has a vortex excitation at $p$. This leaves four sectors to consider, corresponding to $(\Phi_x,\Phi_y)=(++),(+-),(-+),(--)$~[we use $(++)$ as a shorthand for $(+1,+1)$, and similarly for the other three]. These four sectors can be treated in an identical way, which we do in the following.

We first need to find a solution $\{u_{ij}\}$ to Eq.~\eqref{eq:Wpphixphiy}. For the $(++)$ sector, we can simply take  $u_{ij}=+1$ for all links $\langle ij\rangle$. To obtain solutions for the other three vortex-free sectors, %
notice that we can flip the sign of $\Phi_x$ or $\Phi_y$ by flipping the signs of $u_{ij}$ on a large~(i.e. non-contractible) loop of links, without changing the value of any $W_p$. For example, if we flip  all the $z$-links between $x=L_x-1/2$ and $x=0$~(denote this set of links by $Z_{L_x-1/2,0}$), then we can flip the sign of $\Phi_x$ without flipping any of  the $W_p$. Similarly we can flip the sign of $\Phi_y$ by flipping the signs of all the $y$-links between $y=L_y-1/2$ and $y=0$~(denote this set of links by $Y_{L_y-1/2,0}$). In this way, the solution for the sector $(\Phi_x, \Phi_y)$ can be taken as $u_{ij}=1$ for $\langle ij\rangle\notin Z_{L_x-1/2,0} \cup Y_{L_y-1/2,0} $, $u_{ij}=\Phi_x$ for $\langle ij\rangle\in Z_{L_x-1/2,0} $, and $u_{ij}=\Phi_y$ for $\langle ij\rangle\in Y_{L_y-1/2,0}$.  %

The transfer matrix defined in Eq.~\eqref{eq:Tbondfermion} for all these four sectors can be written in a translationally invariant way provided that we use suitable boundary conditions for the Majorana operators.  To this end, we use $i=(\vec{r},\lambda)$ to label lattice sites, where $\vec{r}$ labels the unit cells, and  $\lambda=0,1,2,3$ label the sites in a unit cell, as shown in Fig.~\ref{fig:brickwall-unitcell}. We define $\hat{c}_{(L_x,y),\lambda}=\Phi_x\hat{c}_{(0,y),\lambda}$ and $\hat{c}_{(x,L_y),\lambda}=\Phi_y\hat{c}_{(x,0),\lambda}$, corresponding to periodic or antiperiodic boundary conditions.
Then the transfer matrices for all the four vortex-free sectors have the same expression
\begin{equation}\label{eq:Ttilde}
	\hat{T}'=e^{K_x\sum_{x}i\hat{c}_i\hat{c}_j}e^{K_y\sum_{y}i\hat{c}_i\hat{c}_j}e^{K_z\sum_{z}i\hat{c}_i\hat{c}_j}.
\end{equation}
where the above boundary condition on $\hat{c}_i,\hat{c}_j$ is used, and it is understood that the lattice coordinates of $i,j$ for each link $\langle ij\rangle$ should be consecutive numbers, e.g. the term on a flipped $z$-link is understood as $ \hat{c}_{(L_x-1,y),\lambda} \hat{c}_{(L_x,y),\lambda'}$ instead of $\hat{c}_{(L_x-1,y),\lambda} \hat{c}_{(0,y),\lambda'}$.

The rest of the task is to find the eigenvalues of the translationally invariant vortex-free transfer matrix $\hat{T}'$ in Eq.~\eqref{eq:Ttilde} under the four possible boundary conditions $(++),(+-),(-+),(--)$. To this end, we introduce the Fourier transform  of the Majorana operators
\begin{eqnarray}\label{eq:Ftrans}
  \hat{a}_{\vec{q},\lambda}&=&\frac{1}{\sqrt{2N}}\sum_{\vec{r}} e^{-i\vec{q}\cdot\vec{r}}\hat{c}_{\vec{r},\lambda},\nonumber\\
  \hat{c}_{\vec{r},\lambda}&=&\sqrt{\frac{2}{N}}\sum_{\vec{q}} e^{i\vec{q}\cdot\vec{r}}\hat{a}_{\vec{q},\lambda},
\end{eqnarray}
where $N=L_x L_y$ is the total number of unit cells. 
The quasi-momentum in the $\alpha$-direction $q_\alpha$ is quantized as  $2n\pi/L_\alpha$ where $n\in\mathbb{Z}$ if $\Phi_\alpha=+1$ and $n\in\mathbb{Z}+1/2$ if $\Phi_\alpha=-1$ . The operators $\hat{a}_{\vec{q},\lambda}$ satisfy $\hat{a}_{\vec{q},\lambda}^\dagger=\hat{a}^{\phantom{\dagger}}_{-\vec{q},\lambda}$ and $\{\hat{a}^{\phantom{\dagger}}_{\vec{p},\lambda},\hat{a}^\dagger_{\vec{q},\mu}\}=\delta_{\vec{p},\vec{q}}\delta_{\lambda,\mu}$. We can now rewrite $\hat{T}'$ as
\begin{eqnarray}\label{eq:TtildeFF}
  \hat{T}'&=&\exp\left[2K_x\sum_{\vec{q}}(i\hat{a}_{\vec{q},0}\hat{a}_{-\vec{q},1}+i\hat{a}_{\vec{q},2}\hat{a}_{-\vec{q},3})\right]\nonumber\\
             &&\times\exp\left[2K_y\sum_{\vec{q}}(i\hat{a}_{\vec{q},0}\hat{a}_{-\vec{q},1}e^{iq_y}+i\hat{a}_{\vec{q},2}\hat{a}_{-\vec{q},3}e^{-iq_y})\right]\nonumber\\
  &&\times\exp\left[2K_z\sum_{\vec{q}}(i\hat{a}_{\vec{q},2}\hat{a}_{-\vec{q},1}+i\hat{a}_{\vec{q},0}\hat{a}_{-\vec{q},3}e^{iq_x})\right],\nonumber\\
  &\equiv &\tilde{T}_{0}\prod_{\vec{q}+}\tilde{T}_{\vec{q}},
\end{eqnarray}
where $\tilde{T}_{0}$ contains all the terms with $\vec{q}\equiv -\vec{q}~(\mathrm{mod}~ 2\pi)$, and $\prod_{\vec{q}+}$ is the product over $\vec{q}$ with $\vec{q}\not\equiv -\vec{q}~(\mathrm{mod}~ 2\pi)$ such that each pair $\pm\vec{q}$ appears exactly once, and in the last line we have rearranged terms of different $\vec{q}$ modes using $[\tilde{T}_{\vec{q}},\tilde{T}_{0}]=0$, and $[\tilde{T}_{\vec{q}},\tilde{T}_{\vec{p}}]=0$ for $\vec{q}\neq\pm \vec{p}$. Because of this commutativity, all the $\tilde{T}_{\vec{q}}$ and $\tilde{T}_{0}$ can be simultaneously diagonalized. We treat $\tilde{T}_{\vec{q}}$ first, which can be written as 
\begin{eqnarray}\label{def:Tq}
\tilde{T}_{\vec{q}}&=&e^{2K_x \sum_{\lambda,\mu}P^{(\vec{q})}_{\lambda\mu}\hat{a}^\dagger_{\vec{q}\lambda}\hat{a}^{\phantom{\dagger}}_{\vec{q}\mu}}e^{2K_y \sum_{\lambda,\mu} Q^{(\vec{q})}_{\lambda\mu}\hat{a}^\dagger_{\vec{q}\lambda}\hat{a}^{\phantom{\dagger}}_{\vec{q}\mu}}\nonumber\\
&&\times e^{2K_z \sum_{\lambda,\mu} R^{(\vec{q})}_{\lambda\mu}\hat{a}^\dagger_{\vec{q}\lambda}\hat{a}^{\phantom{\dagger}}_{\vec{q}\mu}},
\end{eqnarray}
where the $4\times 4$ matrices $P^{(\vec{q})},Q^{(\vec{q})},R^{(\vec{q})}$ are~(we drop the superscript $\vec{q}$ when there is no confusion)
\begin{eqnarray}\label{def:PQR}
P&=&\begin{pmatrix}
	0 & i & 0 & 0 \\
	-i& 0 & 0 & 0 \\
	0 & 0 & 0 & i \\
	0 & 0 & -i& 0
\end{pmatrix},
R=\begin{pmatrix}
	0 & 0 & 0 & ie^{-iq_x} \\
	0 & 0 & -i & 0 \\
	0 & i & 0 & 0 \\
	-ie^{iq_x} & 0 & 0 & 0
\end{pmatrix},\nonumber\\
Q&=&\begin{pmatrix}
	0 & ie^{-iq_y} & 0 & 0 \\
	-ie^{iq_y}& 0 & 0 & 0 \\
	0 & 0 & 0 & ie^{iq_y} \\
	0 & 0 & -ie^{-iq_y}& 0
\end{pmatrix}.
\end{eqnarray}
Notice that the fermion bilinears $\hat{a}^\dagger_{\vec{q}\lambda}\hat{a}^{\phantom{\dagger}}_{\vec{q}\mu}$ in Eq.~\eqref{def:Tq} form the basis of an $\mathfrak{sl}(4)$ Lie algebra, so $\tilde{T}_{\vec{q}}$ is an element of the corresponding $\mathrm{SL}(4)$ Lie group. Using the relation between the fundamental representation and the free fermion representation of this Lie algebra and group, (similar to the method in App.~\ref{appen:numerical_method}), one can show that 
\begin{equation}\label{eq:Tq_diagonal}
	\tilde{T}_{\vec{q}}=e^{\epsilon_{\vec{q},1}(\hat{n}_{\vec{q},1}-\hat{n}_{\vec{q},\bar{1}})+\epsilon_{\vec{q},2}(\hat{n}_{\vec{q},2}-\hat{n}_{\vec{q},\bar{2}})},
\end{equation}
where $ e^{\pm\epsilon_{\vec{q},1}}, e^{\pm\epsilon_{\vec{q},2}}$ are the eigenvalues of the matrix $T_{\vec{q}}=e^{2K_x P}e^{2K_yQ}e^{2K_zR}$, which is the representation of $\tilde{T}_{\vec{q}}$ in the fundamental representation of the $\mathrm{SL}(4)$ Lie group, and $\hat{n}_{\vec{q},j},\hat{n}_{\vec{q},\bar{j}}$ ~(with $j\in \{1,2\}$) are mutually commuting fermion number operators. The single mode energies $\epsilon_{\vec{q},j}$ can be analytically calculated by solving the quartic equation $P_{T_{\vec{q}}}(x)=0$, where $P_{T_{\vec{q}}}(x)$ is the degree four characteristic polynomial  of the $4\times 4$ matrix  $T_{\vec{q}}$. This quartic equation can be simplified to a quadratic one $z^2+Az+B=0$, where $z=(x+1/x)/2=\cosh\epsilon_{\vec{q},j}$~(for $j=1,2$), and 
\begin{eqnarray}\label{eq:chepsilon}
A&=&-2c_3(c_1c_2+s_1s_2 \cos q_y),\nonumber\\
B&=&\frac{1}{8}S_1S_2(3+C_3-2s_3^2\cos q_x)\cos q_y+\frac{1}{2} s_1^2 s_2^2\cos(2q_y)\nonumber\\
&&{}+\frac{1}{4}s_3^2(1-C_1C_2)\cos(q_x)+\frac{C_1+C_2+3C_3}{8}\nonumber\\
&&{}+\frac{C_1C_2}{4}+\frac{C_1C_2C_3}{8},
\end{eqnarray}
where $c_j=\cosh 2K_j,s_j=\sinh 2K_j,C_j=\cosh4K_j$, and $S_j=\sinh4K_j$.
Since the eigenvalues of $T_{\vec{q}}$ come in pairs $\pm\epsilon_{\vec{q},1}, \pm\epsilon_{\vec{q},2}$, we can assume without loss of generality that $0\leq \mathrm{Re}[\epsilon_{\vec{q},1}]\leq \mathrm{Re}[\epsilon_{\vec{q},2}]$. Then the maximal eigenvalue of $\tilde{T}_{\vec{q}}$ is $e^{\epsilon_{\vec{q},1}+\epsilon_{\vec{q},2}}$.

The term $\tilde{T}_{0}$ in the last line of Eq.~\eqref{eq:TtildeFF} is defined by %
$\tilde{T}_{0}=\prod_{\vec{q}\equiv \vec{0}~(\mathrm{mod}~\pi)}\tilde{T}_{0,\vec{q}}$ with
\begin{equation}\label{eq:T_0qxqy}
	\tilde{T}_{0,\vec{q}}=e^{2(K_x+K_ye^{iq_y})i\hat{a}_{\vec{q}0}\hat{a}_{\vec{q}1}(1-\hat{P}_{\vec{q}})} e^{2K_zi\hat{a}_{\vec{q}1}\hat{a}_{\vec{q}2}(1+\hat{P}_{\vec{q}}e^{iq_x})},
\end{equation}
where $\hat{P}_{\vec{q}}=4\hat{a}_{\vec{q},0}\hat{a}_{\vec{q},1}\hat{a}_{\vec{q},2}\hat{a}_{\vec{q},3}$. 
Using $\hat{a}^\dagger_{\vec{q},\lambda}=\hat{a}_{\vec{q},\lambda}, \hat{a}_{\vec{q},\lambda}^2=1/2$,
the eigenvalues of $\tilde{T}_{0,\vec{q}}$ can be straightforwardly obtained by diagonalizing  Eq.~\eqref{eq:T_0qxqy}, and one can show that the largest one happens to be equal to $e^{(\epsilon_{\vec{q},1}+\epsilon_{\vec{q},2})/2}$.  %

We have not yet taken into account the fermion parity restriction in Eq.~\eqref{eq:parity_restriction}. However, as we will see in Sec.~\ref{sec:TPloop},
this constraint changes $\ln\Lambda_{\mathrm{max}}$ by at most $O(\epsilon_{\vec{q},j})$, and therefore does not affect the free energy density in the thermodynamic limit. The largest eigenvalue $\Lambda_{\mathrm{max}}^{(\Phi_x,\Phi_y)}$ of $\hat{T}'$ is
\begin{eqnarray}\label{eq:lamdba_max}
	\ln\Lambda_{\mathrm{max}}^{(\Phi_x,\Phi_y)}&=&\frac{1}{2}\sum_{\vec{q}}(\epsilon_{\vec{q},1}+\epsilon_{\vec{q},2}),%
\end{eqnarray}
where $(\Phi_x,\Phi_y)\in\{(++),(+-),(-+),(--)\}$, and the RHS implicitly depends on $(\Phi_x,\Phi_y)$  through the quantization of $\vec{q}$. %
The largest eigenvalue $\Lambda_{\mathrm{max}}$ of $\hat{T}$ is the largest of these four. Regardless of which one is the largest, the free energy density~(per site) in the thermodynamic limit is
\begin{eqnarray}\label{eq:free_energy}
f\equiv\frac{F}{12MN}&=&-\frac{k_BT}{24N}\sum_{\vec{q}}(\epsilon_{\vec{q},1}+\epsilon_{\vec{q},2})\\
&=&-\frac{k_BT}{96\pi^2}\iint_{[-\pi,\pi]^2}(\epsilon_{\vec{q},1}+\epsilon_{\vec{q},2})~d^2q,\nonumber
\end{eqnarray}
where the free energy $F$ is defined in Eq.~\eqref{def:free_energy}.

\subsection{Excitations and phase boundaries}\label{sec:phase_boundary}
\begin{figure}
	\center{\includegraphics[width=0.7\linewidth]{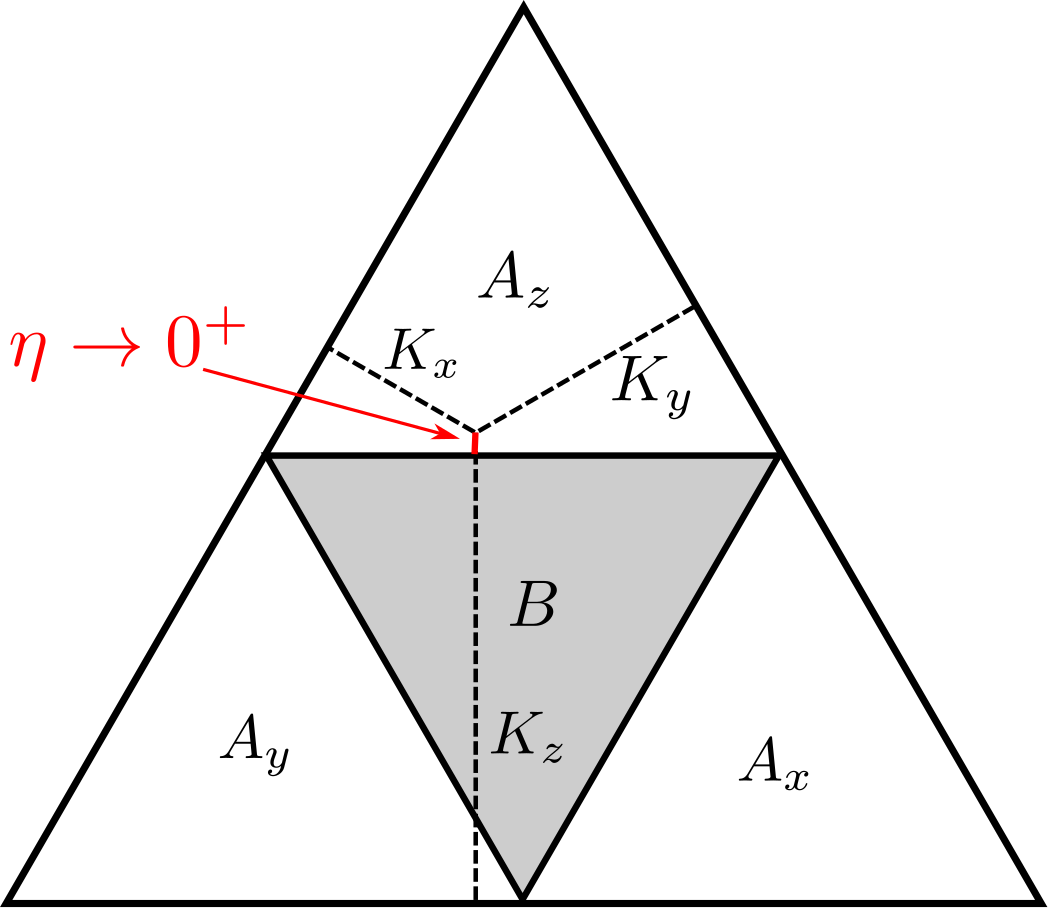}}
	\caption{\label{fig:critical} A 2D section of the 3D phase diagram of our model,  with the intersecting plane $K_x+K_y+K_z=\beta(J_x+J_y+J_z)=\mathrm{const.}$ The parameters $(K_x,K_y,K_z)$ of an arbitrary point in the diagram is given by the distance from that point to the three sides of the triangle. 
	The $B$-phase~(shaded) has a gapless transfer matrix, which acquires a gap after a small perturbation is introduced~(Sec.~\ref{sec:phase_boundary}).
	The $A$-region has a gapped transfer matrix and consists of three disjoint phases $A_x, A_y,A_z$. In Sec.~\ref{sec:critical_exp} we study the critical behavior of the free energy  as we approach the phase boundary  from the $A_z$-phase $\eta\equiv K_z - K_x - K_y \to 0^+$. }
\end{figure}
In this section we study other eigenvalues of the transfer matrix $\hat{T}$ beyond the principal eigenvalue, and, using this, determine the phase diagram of our model. It is useful to define an effective non-Hermitian Hamiltonian 
$	\hat{H}=-\ln \hat{T}.$
In this way the principal eigenstates of $\hat{T}$ are mapped to the ground states of $\hat{H}$ and the eigenvalues $\Lambda_j$ of $\hat{T}$  are related to excitation energies $E_j-E_0$ of $\hat{H}$ by  $E_j-E_0=\ln \Lambda_{\text{max}}-\ln\Lambda_j$. For the rest of this paper, we use the term ``excitation spectrum of $\hat{T}$'' to mean the excitation spectrum of $\hat{H}$, and call the transfer matrix ``gapped''~(``gapless'') if $\mathrm{Re}[E_j-E_0]$ is gapped~(gapless) in the thermodynamic limit. The spectral gap $\Delta=\min_{j\neq 0}\mathrm{Re}[E_j-E_0]$ plays an important role  in the physical properties of the original classical Ising model. First, as we will see in a moment, the phase boundary of our model is determined by regions where $\Delta$ vanishes. Secondly, although we do not calculate in this paper, we claim that two point connected correlations $\langle \sigma_i\sigma_j\rangle_c$~(or more generally, $\langle O_i O_j\rangle_c$ where $O_i$ is a product of classical  spins in a local region) decay exponentially in distance when $\Delta>0$, while there are algebraically decaying correlations when $\Delta=0$. 

There are two types of excitations: fermionic excitations, corresponding to the positive energy eigenmodes of the fermionic transfer matrix $\hat{T}'$, %
and vortex excitations, corresponding to eigenstates of $\hat{T}'$ in a different sector $\{W_p,\Phi_x,\Phi_y\}$ where some of $W_p$s are equal to $-1$. Vortices can only be created in pairs. 
A pair of vortices can be created by first drawing a segment connecting the two vortices~(the segment should avoid passing through lattice sites) and then flipping $u_{ij}$ on all the lattice edges intersecting with this segment~(similar to Kitaev's honeycomb~\cite{Kitaev2006} and toric code~\cite{Kitaev2003Fault} models).  Our analysis in App.~\ref{appen:Lieb} and the numerical results in App.~\ref{appen:numerical_vgap} suggest that the vortices have  gapped and positive excitation energies. %
On the other hand, the fermionic excitations can become gapless for certain values of $(K_x,K_y,K_z)$, and this determines the phase boundary of our model. 

We emphasize that it is the gap closing of the real part of $\epsilon_{\vec{q},1}$  that determines the phase boundary~\footnote{In fact, for real $K_x,K_y,K_z$, the single fermion energies $\epsilon_{\vec{q},1},\epsilon_{\vec{q},2}$ are real; so the distinction between  $\epsilon_{\vec{q},1}$ and $\mathrm{Re}[\epsilon_{\vec{q},1}]$ is unimportant here. In particular, one obtains the same phase diagram even if $\epsilon_{\vec{q},1}=0$ is used as a criterion for phase transition. }. This claim is based on the analysis in App.~\ref{appen:proof_analyticity}, where we rigorously prove that the free energy $f$ defined in Eq.~\eqref{eq:free_energy} is complex analytic in all its parameters when  $\mathrm{Re}[\epsilon_{\vec{q},1}]>0~\forall \vec{q}\in [-\pi,\pi]^2$.  The proof also suggests that when the gap closes $\mathrm{Re}[\epsilon_{\vec{q},1}]=0$, there are branch points in $\epsilon_{\vec{q},1}+\epsilon_{\vec{q},2}$ that leads to non-analytic behavior of $f$, which we calculate directly in Sec.~\ref{sec:critical_exp}.%

We find two distinct phases corresponding to whether $\mathrm{Re}[\epsilon_{\vec{q},1}]$ is gapped or gapless. The phase boundary is determined as follows. One can show that for fixed $q_y$ the minimum of 
$\mathrm{Re}[\epsilon_{\vec{q},1}]$ occurs at $q_x=0$~[since $\partial_{q_x}\epsilon_{\vec{q}}=f(\epsilon_{\vec{q}},q_y) \sin q_x $ for some positive function $f(\epsilon_{\vec{q}},q_y) $]. Furthermore, in the gapped phase the minimum of $\mathrm{Re}[\epsilon_{(0,q_y),1}]$ occurs either at $q_y=0$ or $q_y=\pi$. %
Therefore, the phase transition occurs when $\mathrm{Re}[\epsilon_{\vec{q},1}]$ vanishes at either $\vec{q}=(0,0)$ or $\vec{q}=(0,\pi)$, 
which happens when one of $K_x,K_y,K_z$ equals the sum of the other two~[this can be seen by diagonalizing $\tilde{T}_{\vec{q}}$ in Eq.~\eqref{def:Tq} at $\vec{q}=(0,0)$ or $(0,\pi)$]. When $K_x,K_y,K_z$ form three sides of a triangle~(we call this the $B$-region, shown as the shaded triangle in Fig.~\ref{fig:critical}), the spectrum is gapless, and when one of $K_x,K_y,K_z$  is bigger than the sum of the other two, the spectrum is gapped~(we call this the $A$-region, consisting of three disjoint white triangles in Fig.~\ref{fig:critical}). The phase diagram in terms of $K_x,K_y, K_z$ is shown in Fig.~\ref{fig:critical}, which is identical to the phase diagram of Kitaev's honeycomb model~\cite{Kitaev2006}. 

The fermionic spectrum of the $B$-phase  can be gapped by adding suitable perturbations. For example, we can add small imaginary parts to $J_x, J_y$, so that $K_x\to K_x+i\kappa,K_y\to K_y-i\kappa$, and then add a small real part to the coupling constants of the $x,y$ links that break the lattice reflection symmetry, in the pattern shown in Fig.~\ref{fig:brickwall-unitcell}. Here $\kappa$ is a small real number $|\kappa|\ll |K_{i}|,i=x,y,z$. (Notice that this corresponds to modifying the link coupling constants of the original classical statistical model on all the $x$ and $y$ planes, which breaks the reflection symmetry of the 3D lattice.)  App.~\ref{appen:gap_phases_B} proves that a subregion of the $B$-phase is gapped by this perturbation. More specifically, when $|K_z|/2<|K_x|=|K_y|$, we have $\Delta=\min_{\vec{q}}\mathrm{Re}[\epsilon_{\vec{q},1}]\propto\kappa^2$. This fact will be useful for Sec.~\ref{sec:TPloop} where we calculate the topological degeneracy of $\Lambda_{\max}$ and Sec.~\ref{sec:loop_observables} where we find loop observables whose expectation values distinguish the two phases. Notice that while our proof of Lieb's theorem in App.~\ref{appen:Lieb} assumes real $K_x,K_y,K_z$, as long as the vortices are gapped, the principal eigenstate is still in the vortex-free sector if $\kappa$ is sufficiently small, which we  assume throughout this paper. 

\subsection{Critical exponents}\label{sec:critical_exp}
In this section we study the critical behavior of our model near the phase boundary between the $A$ and $B$ phases, and show that this is a third order phase transition. Specifically, we parameterize the distance to the phase boundary by $\eta = K_z - K_x - K_y$ and show that as the phase boundary is approached from the $A$-phase side, $\eta \to 0^+$, the leading singular part of the free energy is $f \sim \eta^{5/2}$~\footnote{We are approaching the phase boundary strictly inside the big triangle, i.e. the parameters $K_x,K_y,K_z$ are all nonzero. If one instead approaches the point where two phase boundaries meet from along a side of the big triangle, then one can show that the transition is in 2D Ising universality class, where $f \sim \eta^{2}\ln\eta$.}.

We start from the expression in Eq.~\eqref{eq:free_energy}. Near the phase boundary, the leading singular part of $f$ is contributed by the integration near $\vec{q}=(0,\pi)$ where $\epsilon_{\vec{q},1}$ approaches zero. Letting $\vec{q}=(p_x,\pi+p_y)$ where $p_x,p_y\ll 1$, we expand $\epsilon^2_{\vec{q},1}$ in powers of the small parameters $p_x,p_y$, and $\eta$. Using Eq.~\eqref{eq:chepsilon} and $\cosh \epsilon\approx 1+\epsilon^2/2$ for $\epsilon\ll 1$, we have
\begin{eqnarray}
	\epsilon^2_{\vec{q},1}%
	&=& \frac{s^2_3}{4} p_x^2+2 \frac{s_1s_2}{s_3} \eta p_y^2+4\eta^2+\frac{s_1^2s_2^2}{4s_3^2}p_y^4\\
	&&{}+O(p_x^4)+O(p_x^2\eta)+O(p_x^2p_y^2)+O(p_y^4\eta), \nonumber
\end{eqnarray}
where the neglected terms will not affect the leading-order singularity. The leading singular part of $f$ is 
\begin{eqnarray}\label{eq:f_singular}
	f&\sim&-\frac{1}{48\pi^2\beta s_1s_2}\iint \sqrt{ p_x^2+2 \eta p_y^2+4\eta^2+\frac{p_y^4}{4}} ~dp_x dp_y\nonumber\\
	&\sim &\frac{1}{48\pi^2\beta s_1s_2}\int 2(p_y^2+4\eta)^2\ln(p_y^2+4\eta) ~ dp_y\nonumber\\
	&\sim &\frac{64}{45\beta\pi s_1s_2}\eta^{\frac{5}{2}},
\end{eqnarray}
where in the first line we rescale the integration variables $p_x,p_y$, the integration range is a fixed-length interval passing through the origin, say $[-\epsilon,\epsilon]^2$ with $0<\epsilon\ll 1$, and we use $\sim$ to indicate that an unimportant analytic part has been ignored. %
Therefore, the third derivative $\partial^3_\eta f$ diverges as $\eta\to 0^+$, i.e., the phase transition is third order. 
\subsection{Topological degeneracy}\label{sec:TPloop}
In this section we show that the largest eigenvalues of the transfer matrix $\hat{T}$ of our original spin model are topologically degenerate, and the degeneracy depends on the phase. This topological degeneracy gives rise to the topological behaviors of the loop observables presented in the next section.

To this end, we need to compare the values  $\Lambda_{\mathrm{max}}^{(\Phi_x,\Phi_y)}$ of  the four sectors $(\Phi_x,\Phi_y)\in\{(++),(+-),(-+),(--)\}$, given in Eq.~\eqref{eq:lamdba_max}. Let us focus on regions where $\mathrm{Re}[\epsilon_{\vec{q}}]$ is gapped, i.e. the $A$-region and the $B$-region with the perturbation discussed in Sec.~\ref{sec:phase_boundary}. 
In App.~\ref{appen:FSE} we show that the largest eigenvalues
 $\ln\Lambda_{\mathrm{max}}^{(\Phi_x,\Phi_y)}$ of each of the four sectors are equal up to an exponentially small correction $O(e^{-L/\xi})$, where  $\xi$ is a fixed correlation length. %
This suggests a 4-fold topological degeneracy since all the fermion and vortex excitations are gapped. However, we have not taken into account the fermion parity constraint yet. 
As we discussed in Sec.~\ref{sec:map_fermion}, only those eigenstates of $\hat{T}'$ that satisfy the fermion parity constraint Eq.~\eqref{eq:parity_restriction} correspond to eigenstates of $\hat{T}$. So the actual degeneracy of $\hat{T}$ is the number of ``parity-compatible'' sectors, i.e. sectors whose principal eigenstate $|\Lambda_{\mathrm{max}}^{(\Phi_x,\Phi_y)}\rangle$ satisfies the fermion parity constraint.
The fermion parity constraint Eq.~\eqref{eq:parity_restriction}, written in terms of $\hat{a}_{\vec{q}\lambda},\Phi_x,\Phi_y$, becomes
\begin{equation}\label{eq:parity_restriction2}
(-1)^{(L_x-1)L_y}\prod_{\substack{\vec{q}+}}P_{\vec{q},0} P_{\vec{q},1}P_{\vec{q},2}P_{\vec{q},3}\prod_{\vec{q}\equiv -\vec{q}} P_{\vec{q}} 
= \Phi_x^{L_y},
\end{equation}
where $P_{\vec{q},\lambda}=(1-2n_{\vec{q},\lambda})$, $P_{\vec{q}}=4\hat{a}_{\vec{q},0}\hat{a}_{\vec{q},1}\hat{a}_{\vec{q},2}\hat{a}_{\vec{q},3}$ and $\equiv$ is equality $\mathrm{mod}~2\pi$.  As we mentioned above Eq.~\eqref{eq:free_energy}, the principal eigenstate of $\tilde{T}_{\vec{q}}$ in Eq.~\eqref{def:Tq} always has $n_{\vec{q},1}=n_{\vec{q},2}=1$ and $n_{\vec{q},{\bar 1}}=n_{\vec{q},{\bar 2}}=0$, so we have  $P_{\vec{q},0} P_{\vec{q},1}P_{\vec{q},2}P_{\vec{q},3}=+1$ for $\vec{q}\not\equiv -\vec{q}$. Therefore, whether a sector $(\Phi_x, \Phi_y)$ is parity-compatible or not is determined by the values of $P_{\vec{q}}$ where $\vec{q}\equiv -\vec{q}$. 

There are only four possible $\vec{q}$ that can satisfy $\vec{q}\equiv -\vec{q}$~: $(0,0),(0,\pi),(\pi,0),(\pi,\pi)$. For the rest of this section, we assume that  $L_x,L_y$ are both even numbers~[we treat the other cases in App.~\ref{appen:TPD}; the conclusions are the same], in which case these four modes appear in the $(++)$ sector only. This means that Eq.~\eqref{eq:parity_restriction2} is trivially satisfied for the sectors $(+-),(-+),(--)$, i.e. $\hat{T}$ has at least a 3-fold degeneracy. For the $(++)$ sector, Eq.~\eqref{eq:parity_restriction2} becomes $P_{00}P_{0\pi}P_{\pi 0}P_{\pi\pi}=+1$.
The value of $P_{\vec{q}}$ for these four Majorana modes in the principal eigenstate  $|\Lambda_{\mathrm{max}}^{(++)}\rangle$ is determined by maximizing the $\tilde{T}_{0,\vec{q}}$ term in Eq.~\eqref{eq:T_0qxqy}.
It is straightforward to see that $P_{\pi0}=P_{\pi\pi}=-1$, $P_{00}=[K_z>K_x+K_y]$, and $P_{0\pi}=[K_z>|K_x-K_y|]$, where $[S]=+1$ if the statement $S$ is true and $[S]=-1$ otherwise.  In the $A$-phases, $P_{00}, P_{0\pi}$ are both $-1$~(for $A_x, A_y$) or both $+1$~(for $A_z$), so $P_{00}P_{0\pi}P_{\pi 0}P_{\pi\pi}=+1$ and $\hat{T}$ has a 4-fold degeneracy.  In the $B$-phase we have $P_{00}=-1, P_{0\pi}=+1$, so $P_{00}P_{0\pi}P_{\pi 0}P_{\pi\pi}=-1$, i.e. the sector $(++)$ is parity-incompatible, and $\hat{T}$ has a 3-fold degeneracy.

\subsection{Loop Observables}\label{sec:loop_observables}
In this section we compute the thermal expectation value of the family of loop observables $\sigma[\mathfrak{L}_{(\alpha\beta)}]$ defined in Eq.~\eqref{eq:def_loop_observable}, and verify our earlier claim that it is equal to $\pm 1$ for contractible loops,  $0$ for large loops in the $A$-phase, and $1/3$ for large loops in the gapped $B$-phase. 

We begin with a contractible loop $\mathfrak{L}_p$ being an elementary plaquette of the brickwall lattice. Using the transfer matrix method, we find
\begin{eqnarray}\label{eq:loop_transfer_matrix}
	\langle\sigma[\mathfrak{L}_{p,(\alpha\beta)}]\rangle&=&\mathrm{Tr}[\hat{W}_p \hat{T}^M]/\mathrm{Tr}[\hat{T}^M]\nonumber\\%\equiv \langle \hat{W}_p \rangle,
	&\underset{M\to\infty}{=}& \frac{1}{D}\sum^D_{j=1}\langle \Lambda^{(L)}_{\max,j}|\hat{W}_p|\Lambda^{(R)}_{\max,j}\rangle\nonumber\\
	&=&+1. 
\end{eqnarray}
where $(\alpha\beta)\in\{(xy),(yz),(zx)\}$, the sum is over all the $D$-fold degenerate principal eigenstates, $\langle \Lambda^{(L)}_{\max,j}|$ and $|\Lambda^{(R)}_{\max,j}\rangle$ are the left and right principal eigenstates of $\hat{T}$, respectively. The last line of Eq.~\eqref{eq:loop_transfer_matrix} follows from the fact that the principal eigenstates of $\hat{T}$ are eigenstates of the conserved operator $\hat{W}_p$ with eigenvalue $+1$. The value of $\langle\sigma[\mathfrak{L}_{(\alpha\beta)}]\rangle$ on larger contractible loops can be calculated in a similar way, and the result is~(up to a possible minus sign) the expectation value of the product of $\hat{W}_p$ for all the plaquette $p$ enclosed by $\mathfrak{L}$. Since the $\hat{W}_p$ mutually commute and have eigenvalue $+1$ on the principal eigenstates, $\langle\sigma[\mathfrak{L}_{(\alpha\beta)}]\rangle$ is $\pm 1$ for contractible loops.  %

The behavior of $\langle\sigma[\mathfrak{L}_{(\alpha\beta)}]\rangle$ is more interesting on non-contractible loops. For a large loop $\mathfrak{L}_y$ parallel to the $y$-direction, as shown in Fig.~\ref{fig:hc0}, we have
\begin{eqnarray}\label{eq:loop_Ly_transfer_matrix}
	\langle\sigma[\mathfrak{L}_{y,(\alpha\beta)}]\rangle&=&-\mathrm{Tr}[\hat{\Phi}_y \hat{T}^M]/\mathrm{Tr}[\hat{T}^M]\nonumber\\%\equiv \langle \hat{W}_p \rangle,
	&\underset{M\to\infty}{=}& -\frac{1}{D}\sum^D_{j=1}\langle \Lambda^{(L)}_{\max,j}|\hat{\Phi}_y|\Lambda^{(R)}_{\max,j}\rangle.	
\end{eqnarray} 
For $A$-phases, this is
\begin{eqnarray}\label{eq:loop_Ly_A}
	\langle\sigma[\mathfrak{L}_{y,(\alpha\beta)}]\rangle&=&-\frac{\langle \hat{\Phi}_y\rangle_{++}+\langle \hat{\Phi}_y\rangle_{+-}+\langle \hat{\Phi}_y\rangle_{-+}+\langle \hat{\Phi}_y\rangle_{--}}{4}\nonumber\\
	&=&0
\end{eqnarray} 
while for the gapped $B$-phase,
\begin{eqnarray}\label{eq:loop_Ly_B}
	\langle\sigma[\mathfrak{L}_{y,(\alpha\beta)}]\rangle&=&-\frac{\langle \hat{\Phi}_y\rangle_{+-}+\langle \hat{\Phi}_y\rangle_{-+}+\langle \hat{\Phi}_y\rangle_{--}}{3}\nonumber\\
	&=&\frac{1}{3}.
\end{eqnarray} 
The value of $	\langle\sigma[\mathfrak{L}_{x,(\alpha\beta)}]\rangle$ for a large loop $\mathfrak{L}_x$ parallel to the $x$-direction is mapped to $-\langle \hat{\Phi}_x\rangle$ [Eq.~\eqref{eq:phixphiy}] and can be calculated in an identical way, leading to the same result. 
We see that the value of $\langle\sigma[\mathfrak{L}_{(\alpha\beta)}]\rangle$ indeed distinguish between contractible and  non-contractible loops, are always quantized at rational values, and can be used as a (nonlocal) order parameter that distinguishes the phases. 

In order for the topological features to be a universal characteristic of the phase, rather than an accidental property (arising, for example, due to the model's solvability), they must be in some way robust against small, local perturbations. We argue that this is likely the case.
Notice that a local perturbation, e.g. a small real magnetic field term $B\sum_j \sigma_j$, in the original classical Ising model can be mapped to a local perturbation in the transfer matrix in  Eq.~\eqref{eq:def_TM}. The classical loop observables defined in Eq.~\eqref{eq:def_loop_observable} stills maps to the loop operators $\hat{W}_p,\hat{\Phi}_x,\hat{\Phi}_y$, but they no longer commute with the perturbed $\hat{T}$, and when they act on $|\Lambda^{(R)}_{\max,j}\rangle$ they create excitations along the loop. Consequently we expect the expectation value $\langle\sigma[\mathfrak{L}_{(\alpha\beta)}]\rangle$ to decay exponentially in the length of $\mathfrak{L}$. 

However, based on the robustness of the topological phases of the  2D quantum systems~(defined by the transfer matrix $\hat{T}$), we expect that there exists a family of perturbed loop observables~(whose definition depends on the perturbation) that have exactly the same properties shown above. The argument is based on the idea of quasi-adiabatic continuation~\cite{Hastings2005Quasiadiabatic}. For simplicity, let us assume $K_x,K_y,K_z\ll 1$ so that $\hat{T}$ can be approximated as a Hermitian operator. Then Ref.~\cite{Hastings2005Quasiadiabatic} shows that there exists a quasi-local unitary transformation $\hat{U}_\lambda$ that evolves the unperturbed principal eigenstates  to the perturbed ones $|\Lambda^{(R)}_{\max,j}\rangle_\lambda=\hat{U}_\lambda|\Lambda^{(R)}_{\max,j}\rangle_{\lambda=0}$, where $\lambda$ is the strength of the perturbation. [Roughly speaking, $\hat{U}_\lambda$ is a finite-time evolution by a locally-interacting Hamiltonian $\sum_{i}\hat{h}_i$ such that $t\|\hat{h}_i\|=O(\lambda)$, where $t$ is the total time duration.] Then the perturbed loop operators $\hat{U}_\lambda\{\hat{W}_p,\hat{\Phi}_x,\hat{\Phi}_y\}\hat{U}_\lambda^\dagger$ have exactly the same expectation values in the perturbed principal eigenstates as in the unperturbed solvable model shown above. And due to the quasi-locality of $\hat{U}_\lambda$, Lieb-Robinson bounds~\cite{Lieb1972,hastings2010locality} show that these perturbed operators are finite-width~(of order $v_{\text{LR}}t$, where $v$ is the Lieb-Robinson speed) extensions of the unperturbed ones. So we do expect robustness in this sense, essentially the same robustness of loop observables in quantum topological phases.

\section{Physical relevance of complex coupling constants }\label{sec:justification}
Although the complex coupling constants of Eq.~\eqref{cond:exact} appear unphysical, this section argues that the model nevertheless gives insights into genuine physical systems.

Foremost, we expect the general strategy of this paper -- finding 3D classical models whose transfer matrices can be solved using techniques previously applied to solvable 2D  quantum models -- to be a fruitful idea that may lead to a wealth of new solvable models, some of which may have real-valued energy. For example, Refs.~\cite{Chapman2020characterizationof,Ogura2020geometric,elman2020free} have classified families of quantum spin models that can be solved by mapping to free fermions, and these provide a fertile source for new 3D solvable models. 

As an example of this strategy, Sec.~\ref{sec:physical_model} shows that the $A$-phase of our model can be realized in a model with real coupling constants. This provides a physical model showing the topological properties.  As a speculative aside, we also note that this demonstrates that even models with complex-valued couplings may have the same universal physics as real-valued physical models, and thus the former may serve as windows into the latter. 

Additionally, Sec.~\ref{sec:quantum_amplitude} shows two different realizations of the partition function of our complex parameter Ising model in certain dynamical processes of a 3D quantum spin system. Both in principle allow the free energy of our model to be measured experimentally. They suggest that the statistical mechanics of Eqs.~(\ref{eq:HIK},\ref{eq:Z}) gives a solvable model of 3D DQPT~\cite{Heyl_2018} that display topological features. %

\subsection{Realization of $A$-phase in a model with real energy}\label{sec:physical_model}

The $A$-phase can be realized in a physical model with real energies, as we now show. Specifically, the model has a phase that reproduces the $A$-phase's topological properties, that contractible  loops have expectation value $\pm 1$ while noncontractible loops have expectation value $0$. 
 
\begin{figure}
	\center{\includegraphics[width=0.7\linewidth]{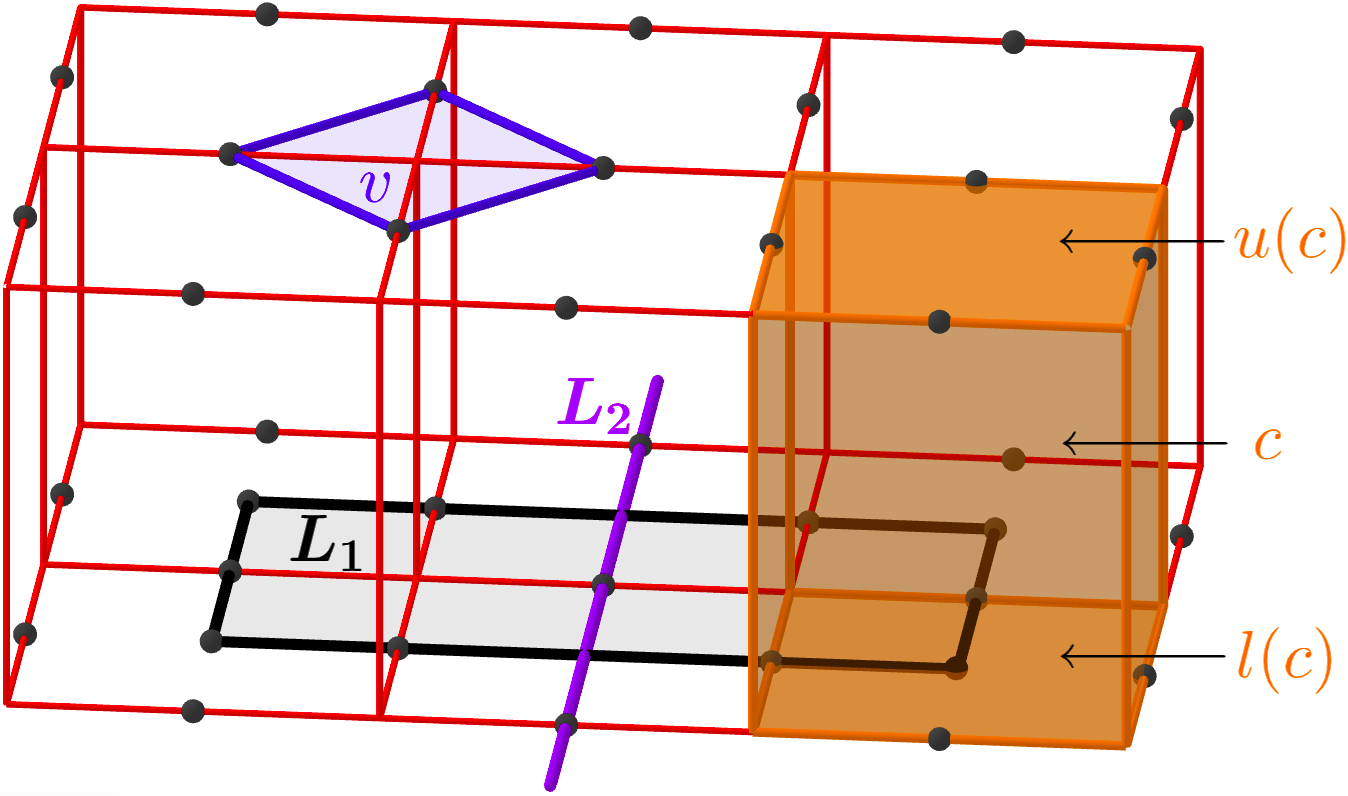}}
	\caption{\label{fig:KTC4} The model Eq.~\eqref{eq:H_KTC4} lies on a 3D cubic lattice where the classical spins sit on links in the $x$ and $y$ directions. There are four spin interactions $\sigma_{v,1}\sigma_{v,2}\sigma_{v,3}\sigma_{v,4}$ between spins around every lattice vertex $v$~(shown as blue diamond) and eight spin interactions $\epsilon[\sigma_{u(c)}, \sigma_{l(c)}]$ around every elementary cube $c$~(shown as orange cube). Example of a contractible loop $L_1$ is shown as black  square, and a noncontractible loop $L_2$ is shown as purple solid line.}
\end{figure}
Consider a 3D square lattice where there is one classical Ising spin on each link in the $x$ and $y$ directions, but no spins live on the links in the $z$ direction, as shown in Fig.~\ref{fig:KTC4}. The energy of a spin configuration $\{\sigma\}$ is given by
\begin{equation}\label{eq:H_KTC4}
E[\{\sigma\}]=-\sum_v \sigma_{v,1}\sigma_{v,2}\sigma_{v,3}\sigma_{v,4}-\sum_{c}\epsilon[\sigma_{\{u(c)\}}, \sigma_{\{l(c)\}}],
\end{equation}
where the first sum is over all vertices $v$, $\sigma_{v,1},\sigma_{v,2},\sigma_{v,3},\sigma_{v,4}$ denote the four spins linked to the vertex $v$, the second sum is over all cubes $c$, and $\{u(c)\},\{l(c)\}$ denote the upper and lower plaquettes of $c$, respectively. We use $\sigma_{\{p\}}=(\sigma_{p,1},\sigma_{p,2},\sigma_{p,3},\sigma_{p,4})$ to denote the configurations of the four spins of the plaquette $p$. The energy of the cube $c$ is defined as $\epsilon[\sigma_{\{u(c)\}}, \sigma_{\{l(c)\}}]=\ln \cosh(1)$ if $\sigma_{\{u(c)\}}=\sigma_{\{l(c)\}}$, $\epsilon[\sigma_{\{u(c)\}}, \sigma_{\{l(c)\}}]=\ln\sinh(1)$ if $\sigma_{\{u(c)\}}=-\sigma_{\{l(c)\}}$ while $\epsilon[\sigma_{\{u(c)\}}, \sigma_{\{l(c)\}}]=-\infty$ otherwise. 

The partition function is 
\begin{equation}\label{eq:Z_KTC4}
Z=\sum_{\{\sigma\}}e^{- E[\{\sigma\}]}=\mathrm{Tr}[\hat{T}^M],
\end{equation}
where the transfer matrix $\hat{T}$ is an operator acting on quantum spins lying on a 2D slice of the lattice, defined by
\begin{equation}\label{eq:T_KTC4}
\hat{T}=\exp\left(\sum_v \hat{\sigma}^z_{v,1}\hat{\sigma}^z_{v,2}\hat{\sigma}^z_{v,3}\hat{\sigma}^z_{v,4}+\sum_{p}\hat{\sigma}^x_{p,1}\hat{\sigma}^x_{p,2}\hat{\sigma}^x_{p,3}\hat{\sigma}^x_{p,4}\right),
\end{equation}
which is simply $e^{-\hat{H}}$ where $\hat{H}$ is the Hamiltonian of Kitaev's toric code model. The principal eigenstates of $\hat{T}$ are the 4-fold degenerate ground states of $\hat{H}$. 

Fig.~\ref{fig:KTC4} shows the family of loop observables we are interested in. Using the same method as in Sec.~\ref{sec:loop_observables}, these classical loop observables can be mapped to the conserved loop operators of the quantum toric code, and
the thermal expectation values of the former are mapped to quantum expectation values of the latter. Averaging over the four topologically degenerate principal eigenstates, we find that the expectation value of contractible loops is $+1$ while non-contractible loops have expectation value $0$. This  reproduces the topological behavior of the $A$-phase of the Ising model presented in  Sec.~\ref{sec:loop_observables}.

\subsection{Realizing the partition function in quantum dynamics}\label{sec:quantum_amplitude}

Another way in which classical statistical models with complex energy can be physically relevant is that the partition function $Z$ can be mapped to measurable quantities of certain~(unitary) quantum dynamical processes in 3D (not 2D) quantum systems. In this section we show two such constructions: Sec.~\ref{sec:transition_amplitude} shows how to realize $Z$ as a transition amplitude, while Sec.~\ref{sec:probe_spin_coherence} shows that $Z$ gives the quantum coherence of a probe spin-$1/2$ coupled to the whole system. The phase transition we studied in our model is then mapped to a DQPT in these quantities.

\subsubsection{Interpreting the partition function as a transition amplitude}\label{sec:transition_amplitude}
Consider a 3D quantum spin system on the same lattice as Fig.~\ref{fig:hc0}, and with a Hamiltonian given by Eq.~\eqref{eq:HIK} with all $\sigma_i$ replaced by $\hat{\sigma}_i^z$, and we will take all the parameters $J_x,J_y,J_z,J_\perp,h$ to be real to guarantee hermiticity. The quantum transition amplitude between two arbitrary states is
\begin{equation}
	\langle A|e^{-it\hat{H}}|B\rangle=\sum_{\{\sigma\}} e^{-it H[\{\sigma\}]}\langle A|\{\sigma\}\rangle\langle\{\sigma\}|B\rangle,
\end{equation}
where on the RHS we inserted a complete set of $\hat{\sigma}^z$ basis states. If the states $|A\rangle,|B\rangle$ are of the following form
\begin{equation}\label{eq:initial_final_states}
	|A\rangle=\bigotimes_{\langle ij\rangle\in \mathbf{X}}|\psi(A_x)\rangle_{ij}\bigotimes_{\langle ij\rangle\in \mathbf{Y}}|\psi(A_y)\rangle_{ij}\bigotimes_{\langle ij\rangle\in \mathbf{Z}}|\psi(A_z)\rangle_{ij},
\end{equation}
where $\bigotimes_{\langle ij\rangle\in \mathbf{X}}$ is over all the red thick $x$-links in  Fig.~\ref{fig:hc0}, and similarly for $\bigotimes_{\langle ij\rangle\in \mathbf{Y}}$ and $\bigotimes_{\langle ij\rangle\in \mathbf{Z}}$, and the local state on each link $\langle ij\rangle$ is defined as $|\psi(A)\rangle_{ij}=\frac{1}{2\sqrt{\cosh 2\mathrm{Re}(A)}}\sum_{\sigma_i,\sigma_j}e^{A\sigma_i\sigma_j}|\sigma_i,\sigma_j\rangle$. Note that Eq.~\eqref{eq:initial_final_states} defines product states since the thick links $\mathbf{X},\mathbf{Y},\mathbf{Z}$ are non-overlapping. Then we have
\begin{equation}\label{eq:transition_amplitude}
	\langle A|e^{-it\hat{H}}|B\rangle=\mathrm{const.}\times Z(K_x,K_y,K_z,it J_\perp,it h),
\end{equation}
where $K_j=A^*_j+B_j+itJ_j,j=x,y,z$. Therefore, when $tJ_\perp\equiv \pi/4~(\mathrm{mod}~ \pi/2),t h\equiv \pi/4~ (\mathrm{mod} ~\pi)$, the transition amplitude is given by the results we derived previously. %

Quantum transition amplitudes, or closely related objects called dynamical partition functions $f(t)\propto-\ln \langle A|e^{-it\hat{H}}|B\rangle$, are the central objects in the study of DQPTs~\cite{heyl2013dynamical,Andraschko2014dynamical,Heyl2014dynamical,Vosk2014dynamical,Heyl2015dynamical,Schmitt2015dynamical,Heyl_2018}.  In this literature, a dynamical phase transition typically referes to a singularity of the dynamical evolution of a physical quantity~[e.g. $f(t)$] at a critical time. In our model, the time is fixed at special values e.g. $t_0=\pi/(4J_\perp)=\pi/(4h)$ to guarantee solvability, and the singularity occurs in $f(t_0)$ %
as we tune the parameters $K_x,K_y,K_z$ across the phase boundary shown in Fig.~\ref{fig:critical}. Although the situation is slightly different, the analogy is clear, and we also expect that if we fix $K_x,K_y,K_z$ to be exactly at the phase boundary, say $K_z=K_x+K_y$, and let the system evolve in time, then there will likely be a singularity in $f(t)$ at $t_0$, i.e. a DQPT in the usual sense. 

Although quantum transition amplitudes are much harder to measure experimentally compared to local observables, there are promising experimental setups~\cite{Jurcevic2017direct,Tian2020observation} that  measure this quantity in relatively small systems, and are capable of  observing signatures of dynamical phase transition. 

\subsubsection{Mapping the partition function to a probe spin coherence}\label{sec:probe_spin_coherence}
We can also realize the partition function as a probe spin coherence, based on the idea of measuring Yang-Lee zeros in the classical Ising model~\cite{wei2012lee,peng2015experimental}. To this end we couple a probe spin-$1/2$ to the whole  3D (quantum) spin system~(bath) shown in Fig.~\ref{fig:hc0}, with probe-bath interaction %
\begin{eqnarray}\label{eq:probe-bath_interaction}
	H_I%
	&=& \hat{\tau}^z\otimes\left(-J_\perp\sum_{\langle ij\rangle\in\boldsymbol{\perp}}\hat{\sigma}^z_{i}\hat{\sigma}^z_{j}+h\sum_i\hat{\sigma}^z_i\right)\nonumber\\
	&=&\frac{1}{2}\hat{\tau}^z \hat{B}
\end{eqnarray}
where $\hat{\tau}^z$ acts on the probe spin, and $J_\perp$ and $h$ are real. The probe spin is initialized in a superposition state $(|\uparrow\rangle+|\downarrow\rangle)/\sqrt{2}$, and the system~(bath) is initially in equilibrium at temperature $T$ with only interactions in the horizontal $x,y,z$ links, described by the canonical ensemble in Eq.~\eqref{eq:Z}  with $J_\perp=h=0$. When we turn on the probe-bath interaction in Eq.~\eqref{eq:probe-bath_interaction}, the thermal
fluctuation of the field $\hat{B}$ induces decoherence of the probe spin~(due to a random phase $Bt$). The probe spin coherence, defined as the ensemble average
of $e^{i\hat{B} t}$, is mapped to~\cite{wei2012lee}
\begin{equation}\label{eq:probe_spin_decoherence}
	L(t)\equiv\langle e^{i\hat{B} t}\rangle=\frac{Z(\beta J_x,\beta J_y, \beta J_z, i J_\perp t,i h t)}{Z(\beta J_x,\beta J_y, \beta J_z,0,0)}.
\end{equation}
Therefore, when $tJ_\perp\equiv \pi/4~(\mathrm{mod}~ \pi/2)$ and $t h\equiv \pi/4~ (\mathrm{mod} ~\pi)$, $L(t)$ is given by our exact solution in Sec.~\ref{sec:solution}~[notice that the denominator of Eq.~\eqref{eq:probe_spin_decoherence} can be calculated easily and has no singularity], and has a topological phase transition when the parameters $J_x,J_y,J_z$ are tuned across the phase boundary in Fig.~\ref{fig:critical}.  
This kind of probe spin coherence has been measured experimentally in an Ising model of 10 spins~\cite{peng2015experimental}. 

\section{Summary and Outlook}\label{sec:summary}
We exactly solved a 3D classical Ising model on a special 3D lattice, which has some of its coupling constants fixed to imaginary values. The solution exploits the special structure of the transfer matrix, which can be mapped to free fermions using a method similar to the solution of Kitaev's honeycomb model. The analytic solution reveals two distinct phases, with a third order phase transition between them. The two phases can be distinguished by measuring the  product of spins on certain loops, the expectation value of which is quantized to certain rational values~($0$, $1$, or $1/3$), depending only on the phase and the topology of the loop. 
We therefore see that the model not only gives insight into interacting many-body systems in 3D,  but that the behavior it shows is particularly interesting: there are phases with topological properties, and a continuous phase transition between them. 

We expect the topological character of the phases to be universal, as discussed in Sec.~\ref{sec:loop_observables}. 
We also expect universality in some other correlations we have not calculated in this paper. For example, the gapless $B$-phase has power-law decaying two-point correlations. For the gapped $B$-phase~(i.e. with the $\kappa$-perturbation introduced in Sec.~\ref{sec:phase_boundary}), if we put the system on a large cylinder~(with axis parallel to the $z$-direction), due to the existence of gapless chiral edge modes on the boundary of the 2D quantum system~(defined by the transfer matrix $\hat{T}$), we expect that the Ising model has power-law decaying correlations on the cylinder boundary even though all two-point correlations in the bulk decay exponentially. We expect the universality in these power-law exponents~(i.e. remain the same when local perturbations are present).

Despite the unphysical complex coupling constants, we described two connections to physical systems. First, the universal long-distance properties of the two phases and the phase transition may be reproduced in a physical 3D system. We demonstrated this by explicitly constructing another 3D classical statistical model with positive Boltzmann weights that has  topological properties identical to the $A$-phase of our 3D Ising model. More speculatively, this suggests that physical systems may have the same universal behavior as models with complex couplings independent of whether the corresponding real-coupling models can be explicitly found or solved. 
We are unsure if the $B$-phase can be realized in a physical classical system, but we expect this to be challenging if at all possible, since Ref.~\cite{Ringel2017Nogothm} suggests the prevalence of sign-problems in a family of closely related phases.    
Second, the partition function of our model can be realized in certain dynamical processes of a 3D quantum spin system, either as a transition amplitude or as a probe spin coherence, allowing the free energy to be experimentally measured in principle, and the phase transitions studied in our model are related to DQPTs in these 3D quantum systems. 

Our model may have other connections to real physical systems beyond the above two. First, when $K_x, K_y, K_z$ are purely imaginary, our transfer matrix $\hat{T}$ in Eq.~\eqref{eq:def_TM} becomes the unitary evolution operator of a periodically driven Kitaev model studied in Ref.~\cite{FloquetKitaev2017}, so our technique of diagonalizing $\hat{T}$ may be useful in studying certain properties of that system. Second, when $K_x, K_y,K_z\to \pm\infty$, $\hat{T}$ becomes a projection operator representing the sequential measurement of $\hat{\sigma}^z_{i}\hat{\sigma}^z_{j}, \hat{\sigma}^y_{i}\hat{\sigma}^y_{j}$, and $\hat{\sigma}^x_{i}\hat{\sigma}^x_{j}$ on all the $z$-, $y$-, and $x$-links, respectively, which is reminiscent of the measurement process of the honeycomb quantum memory code proposed in Ref.~\cite{hastings2021dynamically}.

Our results may also provide hints for constructing a genuinely 3D--i.e. one which does not factorize into decoupled 2D models--classical statistical model with positive Boltzmann weights and a continuous phase transition, a problem that has been studied for more than 60 years but never solved. As one possible direction, we note that our model can be straightforwardly generalized to a large family of solvable 3D classical statistical models, whose transfer matrix is similar to one of the generalized Kitaev models~\cite{Yao2007Exact,Yang2007Mosaic,SI2008Anyonic,Mandal2009Exactly,Yao2009Algebraic,Wu2009Gamma,Ryu2009Three,Tikhonov2010Quantum,Lai2011SU2,Yao2011Fermionic,Barkeshli2015Generalized} that can also be solved by mapping to free fermions. As free-fermion solvable spin models have been systematically classified recently~\cite{Chapman2020characterizationof,Ogura2020geometric,elman2020free}, it is natural to ask if one of them can be promoted to a transfer matrix that corresponds to a physical 3D classical statistical model.

\acknowledgements
Z.W. is especially grateful to Zongping Gong who suggested the idea in Sec.~\ref{sec:probe_spin_coherence}. We also thank Sarang Gopalakrishnan, Bhuvanesh Sundar, and Maxim Olchanyi for helpful discussions. This work was supported in part by the Welch Foundation~(C-1872) and the National Science Foundation~(PHY-1848304).
K.H.'s contribution benefited from discussions at the KITP, which was supported in part by the National Science Foundation under Grant No. NSF PHY-1748958. %

\appendix
\section*{Appendices}
The Appendices contain technical results used throughout our arguments. In App.~\ref{appen:Lieb} we prove a generalization of Lieb's optimal flux theorem which is used in Sec.~\ref{sec:map_fermion} to show that the principal eigenstates of the transfer matrix have no vortices. In App.~\ref{appen:gap_phases_B} we show that the gapless $B$-phase of our model can be gapped by certain perturbations. In App.~\ref{appen:proof_analyticity} we prove some  analytic properties of the fermionic spectrum which are used in determining the phase boundary in Sec.~\ref{sec:phase_boundary}. In App.~\ref{appen:FSE} we show that the splitting of the principal eigenvalue degeneracy of the transfer matrix decays exponentially with system size, which is important for Sec.~\ref{sec:TPloop} and Sec.~\ref{sec:loop_observables}. In App.~\ref{appen:numerical_method} we give a numerical method to calculate the energy of the vortex excitations of the transfer matrix, which helps us confirm that vortices are gapped. In App.~\ref{appen:TPD} we show that the derivation of Sec.~\ref{sec:TPloop} and Sec.~\ref{sec:loop_observables} can be generalized to arbitrary system size $(L_x,L_y)$, leading to the same conclusions. 
\section{Generalization of Lieb's optimal flux theorem to the transfer matrix Eq.~\eqref{eq:Tbondfermion}}\label{appen:Lieb} 
In this section we generalize Lieb's optimal flux theorem~\cite{lieb1994flux} to the free fermion transfer matrix Eq.~\eqref{eq:Tbondfermion} with real parameters $K_x,K_y,K_z$. The goal is to prove that if we fix the magnitude of the coupling constants on each link and allow their signs $u_{ij}$ to vary independently, then the vortex~(flux) configurations that maximize the principal eigenvalue $\Lambda_{\max}$ of $\hat{T}'$  have no vortex~(i.e. have $W_p=+1$ everywhere). For this, it is sufficient to prove that the vortex-free configurations maximize the fermionic partition function  $Z'=\mathrm{Tr}[\hat{T}'^M]$ for any $M$, and then let $M\to\infty$. 

\begin{figure}                                                                                                            
	\center{\includegraphics[width=0.6\linewidth]{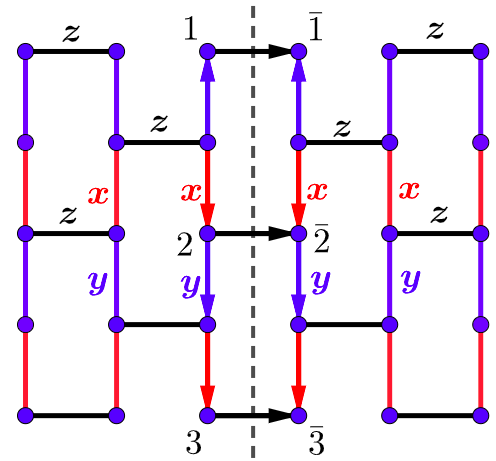}}
	\caption{\label{fig:brickwall} The proof of the generalized Lieb's optimal flux theorem exploits the  reflection symmetry of the brick wall lattice. The mirror line of this reflection is shown as the dashed line cutting through all the $z$-links in a column. The arrows show an example of a reflection symmetric configuration of $\{u_{ij}\}$, where an arrow from $i$ to $j$ means that $u_{ij}=+1=-u_{ji}$. Notice that with the gauge convention $u_{j\bar{j}}=+1$, reflection symmetry of $u_{ij}$ guarantees that all the plaquettes intersecting with the cutting line has zero flux $W_p=+1$, since $u_{ij}$ on the $x$- and $y$-links cancel with their mirror images. The key part of the proof is Eq.~\eqref{eq:Cauchy} which shows that at least one of the optimal flux configurations can be taken to be reflection symmetric with respect to this cutting line. The proof then moves on to apply Eq.~\eqref{eq:Cauchy} to all such cutting lines of the lattice.}
\end{figure}
The proof mostly follows the strategy of Ref.~\cite{lieb1994flux}. The lattice structure is drawn in Fig.~\ref{fig:brickwall}, where both directions are periodic, and we draw a vertical line that cuts the system into two subsystems which are reflections of each other~(up to the difference in the signs of tunneling constants, $u_{ij}$). We will first use reflection positivity~\cite{lieb1994flux} to prove that the optimal flux configuration must have zero flux on the unit cells that intersect with the vertical line, and then apply this conclusion to all such vertical lines~(due to translation invariance in the horizontal direction) to show that the optimal configuration has zero flux everywhere. We use $\hat{R}$ to denote the unitary reflection that maps between the two subsystems, and we denote the links that intersect the cutting line by $1\bar{1}, 2\bar{2},\ldots$, as shown in Fig.~\ref{fig:brickwall}, so that $\hat{R}\hat{c}_j\hat{R}=\hat{c}_{\bar{j}}, j=1,2,3\ldots$. Without loss of generality, we can use the gauge convention in which $u_{j\bar{j}}=1$, since we can always do a gauge transformation on site $j$~(which flips all the $u_{jk}$ linked to $j$) to flip $u_{j\bar{j}}$. We write the free fermion partition function as 
\begin{equation}
	Z=\mathrm{Tr}[(\hat{V}_1 \hat{V}_2\hat{V}_3)^M],
\end{equation}
where $\hat{V}_1=\exp(K_x\sum_{x}u_{ij}i\hat{c}_i \hat{c}_j)$, and similarly for $\hat{V}_2, \hat{V}_3$. Notice that  in both $\hat{V}_1$ and $ \hat{V}_2$, the left and right subsystems are decoupled, so that $\hat{V}_1$ factorizes as $\hat{V}_1=\hat{V}_{1R}\hat{V}_{1L}$, and similarly for $\hat{V}_2$, where we use subscripts $L$~$(R)$ to denote operators acting solely on the subsystem to the left~(right) of the cutting line. $
\hat{V}_3$ involves tunneling between subsystems, and it factorizes as $\hat{V}_3=\hat{V}_{3R}\hat{V}_{3L}\hat{V}_{3I}$, where in our gauge convention mentioned above, 
\begin{eqnarray}
\hat{V}_{3I}&=&\exp(K_z\sum_{j}i\hat{c}_j\hat{c}_{\bar{j}}) \nonumber\\
&=&\cosh(K_z)^m\prod^m_{j=1}(1+t_3 i\hat{c}_j\hat{c}_{\bar{j}}),
\end{eqnarray}
where $t_3=\tanh(K_z)$ and $m$ is the number of cut links~(equal to two times system size in the vertical direction). The partition function becomes
\begin{eqnarray}
	Z&=&\mathrm{Tr}[(\hat{V}_{1R}\hat{V}_{1L}\hat{V}_{2R}\hat{V}_{2L}\hat{V}_{3R}\hat{V}_{3L}\hat{V}_{3I})^M]\\
	 &=&\mathrm{Tr}[(\hat{V}_L \hat{V}_R\hat{V}_{3I})^M]\nonumber\\
	 &=&C\sum_{a_{ij}=0,1}\mathrm{Tr}\left[\hat{V}_L \hat{V}_R (t_3 i\hat{c}_1\hat{c}_{\bar{1}})^{a_{11}}\ldots(t_3 i\hat{c}_m\hat{c}_{\bar{m}})^{a_{1m}}\right.\nonumber\\
	 &&{}\times\left.\hat{V}_L \hat{V}_R (t_3 i\hat{c}_1\hat{c}_{\bar{1}})^{a_{21}}(t_3 i\hat{c}_2\hat{c}_{\bar{2}})^{a_{22}}\ldots(t_3 i\hat{c}_m\hat{c}_{\bar{m}})^{a_{2m}}\right.\nonumber\\
	 &&\cdots\nonumber\\
	 &&{}\times\left.\hat{V}_L \hat{V}_R (t_3 i\hat{c}_1\hat{c}_{\bar{1}})^{a_{M1}}(t_3 i\hat{c}_2\hat{c}_{\bar{2}})^{a_{M2}}\ldots(t_3 i\hat{c}_m\hat{c}_{\bar{m}})^{a_{Mm}}\right],\nonumber
\end{eqnarray}
where $C=(\cosh K_z)^{mM}$, $\hat{V}_L=\hat{V}_{1L}\hat{V}_{2L}\hat{V}_{3L}$, and similarly for $\hat{V}_R$. Our strategy now is to move all the left operators to the left, and all the right operators to the right, without changing the relative order within each class. Notice that $\hat{V}_L$  commute with any operator acting on the right, and $\hat{V}_R$ commute with any operator on the left, while the exchange between $\hat{c}_{i}$ and $\hat{c}_{\bar{j}}$ always introduces a minus sign. We therefore have 
\begin{equation}\label{eq:PF_Z}
	Z[\hat{V}_L,\hat{V}_R]=C\sum_{\mathbf{a}\in \{0,1\}^{Mm}}\mathrm{Tr}[\hat{X}_L^{\mathbf{a}}\hat{X}_R^{\mathbf{a}}](it_3)^{|\mathbf{a}|}(-1)^{c(\mathbf{a})},
\end{equation}
where $\hat{X}_L$ is a product of $M$ number of $\hat{V}_L$ and $|\mathbf{a}|=\sum_{i,j}a_{ij}$ number of $\hat{c}_j$~(suitably ordered) and similarly for $\hat{X}_R$, and $c(\mathbf{a})=|\mathbf{a}|(|\mathbf{a}|-1)/2$ denotes the total number of fermion minus signs introduced by exchanging  $\hat{c}_{i}$ and $\hat{c}_{\bar{j}}$. We write $Z=Z[\hat{V}_L,\hat{V}_R]$ to emphasize the explicit dependence of $Z$ on the flux configuration which determines the signs of the tunneling constants in $\hat{V}_L,\hat{V}_R$.

The next step is to factorize the trace of whole system as a product of traces of subsystems. One way to do this is to combine the two Majorana operators on each $x$-link into a Dirac fermion~(the trace is independent of the choice of the Dirac fermion basis, since different basis are related by a unitary transformation), so that $\mathrm{Tr}[\hat{A}_L\hat{B}_R]=\mathrm{Tr}_L[\hat{A}_L]\mathrm{Tr}_R[\hat{B}_R]$, where $\mathrm{Tr}_L$~($\mathrm{Tr}_R$) denotes the trace on the left~(right) subsystem, and $\mathrm{Tr}_L[\hat{A}_L]=0$ if $\hat{A}_L$ is an odd product of fermionic operators, and similarly for $\hat{B}_R$. Therefore in  Eq.~\eqref{eq:PF_Z} we can restrict the summation to those $\mathbf{a}$ for which $|\mathbf{a}|$ is even, in which case we have $i^{|\mathbf{a}|}(-1)^{c(\mathbf{a})}=1$. Furthermore, we can show that $\mathrm{Tr}_L[\hat{X}_L]$ is real~(and similarly for $\mathrm{Tr}_R[\hat{X}_R]$), i.e. $\mathrm{Tr}_L[\hat{X}_L]=\mathrm{Tr}_L[\hat{X}_L]^*$, since complex conjugation sends $i$ to $-i$ and reverse the signs of all the Majorana fermions on even sites~(leaving Majorana operators on odd sites unchanged), thereby leaving $\hat{X}_L$ invariant. We now have 
\begin{eqnarray}\label{eq:Cauchy}
	Z[\hat{V}_L,\hat{V}_R]^2&=&\left(C\sum_{\mathbf{a}}t_3^{|\mathbf{a}|}\mathrm{Tr}_L[\hat{X}_L^{\mathbf{a}}]\mathrm{Tr}_R[\hat{X}_R^{\mathbf{a}}]\right)^2\nonumber\\
	&\leq&\left(C\sum_{\mathbf{a}}t_3^{|\mathbf{a}|}\mathrm{Tr}_L[\hat{X}_L^{\mathbf{a}}]^2\right)\left(C\sum_{\mathbf{a}}t_3^{|\mathbf{a}|}\mathrm{Tr}_R[\hat{X}_R^{\mathbf{a}}]^2\right)\nonumber\\
	&=&\left(C\sum_{\mathbf{a}}t_3^{|\mathbf{a}|}\mathrm{Tr}_L[\hat{X}_L^{\mathbf{a}}]\mathrm{Tr}_R[\hat{R}\hat{X}_L^{\mathbf{a}}\hat{R}]\right)\nonumber\\
	&&\times\left(C\sum_{\mathbf{a}}t_3^{|\mathbf{a}|}\mathrm{Tr}_L[\hat{R}\hat{X}_R^{\mathbf{a}}\hat{R}]\mathrm{Tr}_R[\hat{X}_R^{\mathbf{a}}]\right)\nonumber\\
	&=&Z[\hat{V}_L,\hat{R}\hat{V}_L\hat{R}]Z[\hat{R}\hat{V}_R\hat{R},\hat{V}_R],
\end{eqnarray}
where in the second line we use the Cauchy-Schwartz inequality. Eq.~\eqref{eq:Cauchy} means that for any flux configuration determined by $[\hat{V}_L,\hat{V}_R]$, at least one of the reflection symmetric configurations corresponding to $[\hat{V}_L,\hat{R}\hat{V}_L\hat{R}]$ or $[\hat{R}\hat{V}_R\hat{R},\hat{V}_R]$ have smaller or equal free energy~(larger or equal $Z$). Notice that each of these reflection symmetric configurations has zero flux on cells intersecting with the cutting line. Therefore the optimal flux configuration~(in case of degeneracy, consider the optimal flux configuration with least $\pi$-fluxes) must have zero flux everywhere, since otherwise we can use Eq.~\eqref{eq:Cauchy} to construct another flux configuration that has either strictly smaller free energy or strictly less $\pi$-fluxes.

The generalized Lieb's theorem shows that at least one of the principal eigenstates have all $W_p$ equal to one. This is also confirmed by the numerical results presented in App.~\ref{appen:numerical_method}, %
which additionally suggests that the excitation energies of vortices remain gapped in the thermodynamic limit. 

\section{The fermion gap of the  $B$-phase}\label{appen:gap_phases_B}
In Sec.~\ref{sec:phase_boundary} we claimed that a subregion of the $B$-phase can be gapped by 
adding small imaginary parts to $J_x, J_y$, so that $K_x\to K_x+i\kappa,K_y\to K_y-i\kappa$, and then adding a small real part to the coupling constants of the $x,y$ links that break the lattice reflection symmetry, in the pattern shown in Fig.~\ref{fig:brickwall-unitcell}. In the following we verify this claim at the line $K_x=K_y$, and show that the fermion gap $\Delta\propto\kappa^2$ in the limit of small $\kappa$. 

Since the operators in the exponential of the fermionic transfer matrix remain quadratic in the Majorana fermion operators with this perturbation, the method used in Sec.~\ref{sec:solve_fermion_TM} still works. We can simply repeat the derivations in Eqs.~(\ref{eq:Ttilde}-\ref{eq:chepsilon}), the only modification now is that
$ e^{\pm\epsilon_{\vec{q},1}}, e^{\pm\epsilon_{\vec{q},2}}$ are the eigenvalues of the modified matrix
\begin{equation}\label{eq:appen_def_Tq}
	T_{\vec{q}}=e^{2K'_x P+2\kappa P_0}e^{2K'_yQ+2\kappa Q_0}e^{2K_zR},
\end{equation}
where $K_x'=K_x+i\kappa$,  $K_y'=K_y-i\kappa$, the matrices $P,Q,R$ are the same as defined in Eq.~\eqref{def:PQR}, and
\begin{eqnarray}\label{def:PQR_ikappa}
	P_0&=&\begin{pmatrix}
		0 & i & 0 & 0 \\
		-i& 0 & 0 & 0 \\
		0 & 0 & 0 & -i \\
		0 & 0 & i& 0
	\end{pmatrix},\nonumber\\
	Q_0&=&\begin{pmatrix}
		0 & ie^{-iq_y} & 0 & 0 \\
		-ie^{iq_y}& 0 & 0 & 0 \\
		0 & 0 & 0 & -ie^{iq_y} \\
		0 & 0 & ie^{-iq_y}& 0
	\end{pmatrix}.
\end{eqnarray}
The eigenvalue problem of $T_{\vec{q}}$ can still be simplified to a quadratic equation $z^2+Az+B=0$, where $z=(x+1/x)/2=\cosh\epsilon_{\vec{q},j}$~(for $j=1,2$).  The expressions of $A,B$ are way more complicated than in Eq.~\eqref{eq:chepsilon}, so we don't show them here. 

To determine the spectral gap of $\epsilon_{\vec{q},1}$~(the one of $\epsilon_{\vec{q},1},\epsilon_{\vec{q},2}$ with smaller real part), we first let $\kappa=0$ and find the $\vec{q}^*$ at which $\epsilon_{\vec{q},1}$ vanishes. This takes the form $\vec{q}^*=(0,q_y^*)$ since $\epsilon_{\vec{q},1}$ is smallest at $q_x=0$ for a fixed $q_y$. %
Requiring that $T_{\vec{q}}$ has an eigenvalue 1 at $\vec{q}^*$, which is equivalent to $1+A+B=0$, we find that 
\begin{equation}
	\cos q_y^*=\frac{\cosh{2K_z}-\cosh 2K_x\cosh 2K_y}{\sinh 2K_x\sinh 2K_y}.
\end{equation}
We can now study the spectrum near the point $\vec{q}^*$, by expanding  the equation $Q(\kappa,q_x,q_y,z)=z^2+Az+B=0$ with  $z=1+\epsilon_{\vec{q},1}^2/2$, $\vec{q}=\vec{q}^*+(q_x,\delta q_y)$. We find that 
\begin{eqnarray}\label{eq:epsilon_expansion_kappa}
	4s_3^2\epsilon^2_{\vec{q},1}&=&\frac{1}{2}Q_{\kappa\kappa}\kappa^2+\frac{1}{2}Q_{xx} q_x^2+\frac{1}{2}Q_{yy}\delta q_y^2+Q_{x\kappa}q_x\kappa\nonumber\\
	&& {}+O(\kappa^4)+O(\kappa^2\delta q_y)+O(q_x\delta q_y\kappa),
\end{eqnarray}
where
\begin{eqnarray}\label{eq:epsilon_expansion_kappa_coefs}
	Q_{\kappa\kappa}&=&32(1-\cos q_y^*)[C_3-\cos q_y^*-C_{1-2}(1-\cos q_y^*)]\nonumber\\
	Q_{xx}&=&s_3^2(4c_1c_2c_3-2-C_1-C_2)\nonumber\\
	Q_{x\kappa}&=&4s_3^2(S_1+S_2)\sin q_y^*\nonumber\\
	Q_{yy}&=&8s_1^2s_2^2\sin^2 q_y^*,
\end{eqnarray}
where $C_{1-2}=\cosh(4K_x-4K_y)$, and $c_j=\cosh 2K_j,s_j=\sinh 2K_j,C_j=\cosh4K_j,S_j=\sinh4K_j$, for $j=x,y,z$~[same as defined in the main text below Eq.~\eqref{eq:chepsilon}]. 

We now determine the  $(q_x,\delta q_y)$ that minimizes the RHS of Eq.~\eqref{eq:epsilon_expansion_kappa}. 
At the line $K_x=K_y$, one can check that $Q_{x\kappa}^2=Q_{\kappa\kappa}Q_{xx}$, and $Q_{xx}>0,Q_{\kappa\kappa}>0$, so the minimum is at $q_x=\kappa Q_{x\kappa}/Q_{xx}+O(\kappa^2),\delta q_y=O(\kappa^2)$. Near this point, in the RHS of Eq.~\eqref{eq:epsilon_expansion_kappa}, terms of order $\kappa^2$ exactly cancel, leaving $\mu\kappa^4$ for some constant $\mu>0$~(the analytic expression for $\mu$ is quite complicated, so we do not show it here). 
Therefore we have $\Delta=\min_{\vec{q}}\epsilon_{\vec{q},1}\propto\kappa^2$. This result is also verified numerically. 

(Notice that when $\kappa=0$, $\epsilon_{\vec{q},1},\epsilon_{\vec{q},2}$  are real; furthermore, since the coefficients in Eq.~\eqref{eq:epsilon_expansion_kappa_coefs} are all real, $\epsilon_{\vec{q},1}$ must be real at order $\kappa^2$, so the distinction between  $\epsilon_{\vec{q},1}$ and $\mathrm{Re}[\epsilon_{\vec{q},1}]$ is unimportant here--the gap for  $\mathrm{Re}[\epsilon_{\vec{q},1}]$ is also proportional to $\kappa^2$.)
\section{The analyticity of $\epsilon_{\vec{q},1}+\epsilon_{\vec{q},2}$}\label{appen:proof_analyticity}
In this section we study the complex analyticity of  $\epsilon_{\vec{q},1}+\epsilon_{\vec{q},2}$ as a function of all its parameters $K_x,K_y,K_z, \kappa, q_x, q_y$. Here $\epsilon_{\vec{q},1},\epsilon_{\vec{q},2}$ are the two eigenvalues of $T_{\vec{q}}=T(K_x,K_y,K_z, \kappa, q_x, q_y)$ defined in Eq.~\eqref{eq:appen_def_Tq}, with $0\leq \mathrm{Re}[\epsilon_{\vec{q},1}]\leq \mathrm{Re}[\epsilon_{\vec{q},2}]$~(if the real parts are equal, order by their imaginary parts). Notice that even though $K_x,K_y, \kappa$ are assumed real in the definition $K_x'=K_x+i\kappa$,  $K_y'=K_y-i\kappa$, we still consider the analytic continuation of $T(K_x,K_y,K_z, \kappa, q_x, q_y)$ to the complex regions. 
This analyticity is used in Sec.~\ref{sec:phase_boundary} in determining the phase boundary, and will also be used in App.~\ref{appen:FSE} in proving the finite size splitting of degenerate $\ln \Lambda_{\text{max}}$  in gapped phases. %

We prove the following theorem:
\begin{theorem}\label{thm1}
	For a given set of $(K_{x0},K_{y0},K_{z0},\kappa_0)\in\mathbb{R}^4$, if $\mathrm{Re}[\epsilon_{\vec{q},1}]>0$ and $\mathrm{Re}[e^{\epsilon_{\vec{q},1}+\epsilon_{\vec{q},2}}]\geq 0$ for all $\vec{q}\in [-\pi,\pi]^2$, then there exists $\rho>0$ such that $\epsilon_{\vec{q},1}+\epsilon_{\vec{q},2}$ is a single-valued complex analytic function~(in all its parameters) in the region %
	\begin{eqnarray}
	R_\rho&=&\left\{(K_x,K_y,K_z, \kappa, q_x, q_y)\in \mathbb{C}^6~|~ \mathrm{Re}[q_i]\in [-\pi,\pi], \right.\nonumber\\
	&&{}|\mathrm{Im}[q_i]|\leq \rho,i=x,y, |\kappa-\kappa_0|\leq\rho,\nonumber\\
	&&|K_j-K_{j0}|\leq\rho,j=x,y,z	\left. \right\}.
	\end{eqnarray}
	
\end{theorem}
\begin{proof}
We use the notation and the results of Sec.~\ref{sec:solve_fermion_TM} and App.~\ref{appen:gap_phases_B}. 
We begin by noticing that the characteristic polynomial $P_{T_{\vec{q}}}(x)$ of $T_{\vec{q}}$ has coefficients  complex analytic in $K_x,K_y,K_z,\kappa,q_x,q_y$ everywhere~(except at infinity), since taking exponentials or determinants of matrices cannot introduce singularities. It follows that the coefficients of $z^2+Az+B=0$ are complex analytic everywhere. Denote the roots by $z_{j}=(x_{j}+1/x_{j})/2=\cosh\epsilon_{\vec{q},j}$ for $j=1,2$. 
Vieta's relations guarantee that all symmetric polynomials of $(z_1,z_2)$, such as $z_1+z_2,z_1z_2,z_1^2+z_2^2$, are polynomials of $A,B$ and therefore analytic everywhere in $\mathbb{C}^6$. 

We now prove that there exists $\rho>0$ such that $\mathrm{Re}[\epsilon_{\vec{q},1}]>0$ in $R_\rho$. %
First, $\mathrm{Re}[\epsilon_{\vec{q},1}]=\ln |x_1|$ is continuous in $K_x,K_y,K_z,\kappa,q_x,q_y$ everywhere, which follows from the continuity of the roots $\{x_j\}_{1\leq j\leq 4}$ of the polynomial $P_{T_{\vec{q}}}(x)$ as a function of its coefficients, and the fact that the roots are ordered by their norm.  We now invoke the theorem that if a  function is continuous on a closed and bounded region, then it is bounded~(and attains its bounds) and uniformly continuous in this region. Since $R_0$ is closed and bounded, let $\Delta>0$ be the minimum of $\mathrm{Re}[\epsilon_{\vec{q},1}]$ in $R_0$. Since $\mathrm{Re}[\epsilon_{\vec{q},1}]$  is uniformly continuous in the closed and bounded region $R_{1.0}$, there exists $\rho\in (0,1.0]$ such that 
\begin{eqnarray}
	&&|\mathrm{Re}[\epsilon_{\vec{q},1}]|_{K_x,K_y,K_z,\kappa,\vec{q}}-\mathrm{Re}[\epsilon_{\vec{q},1}]|_{K_{x0},K_{y0},K_{z0},\kappa_0,\mathrm{Re}[\vec{q}]}|\nonumber\\
	&&\leq \Delta/2,~~~\forall (K_x,K_y,K_z,\kappa,q_x,q_y)\in R_\rho,
\end{eqnarray}
which implies that $\mathrm{Re}[\epsilon_{\vec{q},1}]|_{K_x,K_y,K_z,\kappa,\vec{q}}\geq \Delta/2>0$ in $R_\rho$.

We now study the analyticity of $x_1x_2=e^{\epsilon_{\vec{q},1}+\epsilon_{\vec{q},2}}$. We have
\begin{eqnarray}\label{eq:analytic_x1x2}
	x_1 x_2&=& (z_1+\sqrt{z_1^2-1})(z_2+\sqrt{z_2^2-1})\nonumber\\
	&=&z_1z_2+\sqrt{z_2^2-1}\sqrt{z_1^2-1}\nonumber\\
	&&+(z_1\sqrt{z_2^2-1}+z_2\sqrt{z_1^2-1}).
\end{eqnarray}
Notice that by Vieta's relations, each term in the RHS can be expressed as an algebraic function of $A$ and $B$, and therefore $x_1x_2$ can at most contain branch cuts or branch points in its parameters $(K_x,K_y,K_z,\kappa,q_x,q_y)$. However, in $R_\rho$, we have proved that the four roots $x_2,x_1,1/x_1,1/x_2$ of $P_{T_{\vec{q}}}(x)$ satisfy $|x_2|\geq |x_1|>|1/x_1|\geq |1/x_2|$. Again by the continuity of roots of a polynomial as a function of its parameters, $x_1x_2$ must be a continuous, single-valued function of $A,B$ in the region $R_\rho$~\footnote{Notice the importance of the condition $|x_1|>|1/x_1|$: if  there is a point in $R_\rho$ where $|x_1|=|1/x_1|$, then at this point $x_1x_2$ would jump to $x_2/x_1$ and therefore be discontinuous. }. This rules out any branch cuts or branch points, and therefore $x_1x_2=e^{\epsilon_{\vec{q},1}+\epsilon_{\vec{q},2}}$ must be analytic in $R_\rho$.

We now discuss the analyticity of $\epsilon_{\vec{q},1}+\epsilon_{\vec{q},2}=\ln x_1x_2$  in $R_\rho$.
We already know that $|x_1x_2|> 1$ in $R_\rho$, and the branch point of $\ln x_1x_2$ is at the origin, so we only need to guarantee that, when the parameters $K_x,K_y,K_z,\kappa,q_x,q_y$ vary in $R_\rho$, the values of $x_1x_2$ on the complex plane do not wind around the origin. We already know that $\mathrm{Re}[x_1x_2]\geq 0$ in $R_0$~(by assumption of the theorem), and $x_1x_2$ is continuous in the closed and bounded region $R_\rho$. Therefore  $x_1x_2$ is bounded and uniformly continuous in $R_\rho$. Using a similar method as above, there exists $\rho'\in (0,\rho]$ such that $\mathrm{Re}[x_1x_2]>-1/2$ in $R_{\rho'}$. Combined with $|x_1x_2|> 1$ in $R_{\rho'}\subset R_\rho$, we know that the value set of $x_1x_2$ cannot wind around the origin for $(K_x,K_y,K_z,\kappa,q_x,q_y)\in R_\rho'$.  Therefore $\epsilon_{\vec{q},1}+\epsilon_{\vec{q},2}=\ln x_1x_2$ is a single-valued complex analytic function in $R_{\rho'}$. This concludes the proof.
\end{proof}
We finally remark on the role of Thm.~\ref{thm1} in determining the phase boundary of our model. Since the free energy $f$ is related to $\epsilon_{\vec{q},1}+\epsilon_{\vec{q},2}$ in Eq.~\eqref{eq:free_energy} by an integration in $\vec{q}$ over $[-\pi,\pi]^2$, Thm.~\ref{thm1} is strong enough to guarantee that $f$ is complex  analytic in an open neighborhood of $(K_{x0},K_{y0},K_{z0},\kappa_0)$, if at this point $\mathrm{Re}[\epsilon_{\vec{q},1}]>0$   and  $\mathrm{Re}[e^{\epsilon_{\vec{q},1}+\epsilon_{\vec{q},2}}]\geq 0$ for all $\vec{q}\in [-\pi,\pi]^2$. But we have numerically checked that $\mathrm{Re}[e^{\epsilon_{\vec{q},1}+\epsilon_{\vec{q},2}}]\geq 0$ is almost always satisfied, at least for a wide range of parameters $(K_{x0},K_{y0},K_{z0},\kappa_0)\in \mathbb{R}^4$. So a phase transition can only happen when $\mathrm{Re}[\epsilon_{\vec{q},1}]$ becomes gapless.

\section{Finite size splitting of degenerate $\ln \Lambda_{\text{max}}$ in gapped phases is exponentially small in system size}\label{appen:FSE}
In this section we prove that in the regions where  $\mathrm{Re}[\epsilon_{\vec{q},i}]$ are gapped, the finite size differences among the four different boundary conditions of
\begin{eqnarray}
\epsilon(L)\equiv \frac{1}{L^2}\sum_{\vec{q}}(\epsilon_{\vec{q},1}+\epsilon_{\vec{q},2})
\end{eqnarray} 
decays exponentially in system size $L$. %
In the following we will prove that $|\epsilon(L)-\epsilon_\infty|\leq C e^{-\rho L}$ for some  positive constants $C,\rho$. For simplicity we focus on the double periodic boundary condition~(++), and other cases can be treated in a similar way.

Denote $\epsilon_{\vec{q}}\equiv (\epsilon_{\vec{q},1}+\epsilon_{\vec{q},2})$, and define
\begin{eqnarray}\label{def:fx}
f_L(\vec{x})&=&\frac{1}{L^2}\sum_{\vec{q}} \epsilon_{\vec{q}} e^{i\vec{q}\cdot\vec{x}},\nonumber\\ 
f(\vec{x})\equiv f_\infty(\vec{x})&=&\frac{1}{4\pi^2}\int^\pi_{-\pi}\epsilon_{\vec{q}} e^{i\vec{q}\cdot\vec{x}} d^2q,
\end{eqnarray}
Notice that $f_L(0)=\epsilon(L)$. We have
\begin{eqnarray}
\sum_{m,n\in\mathbb{Z}}f[\vec{x}+(m,n)L]&=&\int\frac{d^2q}{4\pi^2}\sum_{m,n\in\mathbb{Z}}e^{iq_x mL+i q_y nL+i\vec{q}\cdot\vec{x}}\epsilon_{\vec{q}}\nonumber\\
&=&\sum_{r,s\in\mathbb{Z}}\int d^2q\delta[\vec{q}L-2\pi (r,s)]e^{i\vec{q}\cdot\vec{x}}\epsilon_{\vec{q}}\nonumber\\
&=&\frac{1}{L^2}\sum_{\vec{q}}e^{i\vec{q}\cdot\vec{x}}\epsilon_{\vec{q}}\nonumber\\
&=&f_L(\vec{x}).
\end{eqnarray}
Therefore 
\begin{equation}\label{eq:delta_eps_L}
\epsilon_L-\epsilon_\infty=f_L(\vec{0})-f(\vec{0})=\sum_{(m,n)\in\mathbb{Z}^2{\backslash\{(0,0)\}}}f[(m,n)L].
\end{equation}

App.~\ref{appen:proof_analyticity} proved that when $\mathrm{Re}[\epsilon_{\vec{q},i}]$ are gapped and positive, $\epsilon_{\vec{q}}$ is complex analytic in the region $|\mathrm{Im}(q_x)|\leq \rho,|\mathrm{Im}(q_y)|\leq \rho$ for some $\rho>0$. This leads to the exponential decay of
$f(\vec{x})$  in $x$, since 
\begin{eqnarray}
|f(\vec{x})|&=&\frac{1}{4\pi^2}\left|\int^\pi_{-\pi}\epsilon_{\vec{q}} e^{i\vec{q}\cdot\vec{x}} d^2q\right|,	\nonumber\\
&=&\frac{1}{4\pi^2}\left|\int^\pi_{-\pi}\epsilon_{\vec{q}+i\rho(\text{sgn}(x),\text{sgn}(y))} e^{i\vec{q}\cdot\vec{x}} e^{-\rho(|x|+|y|)} d^2q\right|,	\nonumber\\
&\leq& e^{-\rho (|x|+|y|)} \max_{\vec{q}}|\epsilon_{\vec{q}+i\rho(\text{sgn}(x),\text{sgn}(y))}|.
\end{eqnarray}
Then, Eq.~\eqref{eq:delta_eps_L} implies that $|\epsilon(L)-\epsilon_\infty|\leq C e^{-\rho L}$ for some constant $C$.

\section{Numerical solution for vortex sectors without translation invariance}\label{appen:numerical_method}
In Sec.~\ref{sec:phase_boundary} we claimed that vortices are ``gapped'' in the thermodynamic limit for all nonzero $K_1, K_2, K_3$. More precisely, this means that the principal eigenvalue $\ln \Lambda_{\mathrm{\max}}$ of the fermionic transfer matrix $\hat{T}'$ for any sector with vortices is smaller than that of the vortex-free sector by a finite amount $\Delta>0$. This finite excitation  gap is essential for the analysis of topological degeneracy and loop observables in Sec.~\ref{sec:TPloop}. Although the generalized Lieb's theorem in App.~\ref{appen:Lieb} proves that vortices have non-negative excitation energy, we still need to verify that this excitation energy does not approach zero in the thermodynamic limit. To verify the finite excitation gap claim, we need to numerically solve the eigenvalues of $\hat{T}'$, since vortices break translation symmetry and the Fourier transform in the main text cannot be used anymore. 
In the following we first describe the method in App.~\ref{appen:num_method} and then present the result in App.~\ref{appen:numerical_vgap}.
\subsection{Method}\label{appen:num_method}
In the following we present a numerical method to calculate the largest eigenvalue of the free fermion transfer matrix of the form
\begin{eqnarray}
\hat{T}&=&\exp\left(\sum_{i,j}P_{ij} i \hat{c}_i\hat{c}_j\right)\exp\left(\sum_{i,j} Q_{ij} i \hat{c}_i\hat{c}_j\right)\nonumber\\
&&{}\times\exp\left( \sum_{i,j}R_{ij} i \hat{c}_i\hat{c}_j\right),
\end{eqnarray}
where $P,Q,R$ are general $2N\times 2N$ antisymmetric matrices~(not necessarily translationally invariant), and in this section repeated indices indicate summation. Denote by $\mathfrak{so}(2N)$ the Lie algebra of all $2N\times 2N$ antisymmetric matrices. %
For any $X\in \mathfrak{so}(2N)$, define
\begin{equation}
\rho(X)\equiv \frac{1}{4}\sum_{i,j} X_{ij} \hat{c}_i\hat{c}_j.
\end{equation}
It is straightforward to verify that $\rho$ is a representation of $\mathfrak{so}(2N)$, i.e.
\begin{equation}
[\rho(X),\rho(Y)]=\rho([X,Y]).
\end{equation}
We can therefore extend $\rho$ to the corresponding elements of the $SO(2N)$ Lie group by $\rho(e^X)\equiv e^{\rho(X)}$. Notice that $\hat{T}$ is an element of this Lie group in the Majorana fermion representation
\begin{eqnarray}
\hat{T}&=&e^{4i\rho(P)}e^{4i\rho(Q)}e^{4i\rho(R)}\nonumber\\
&=&\rho(e^{4iP})\rho(e^{4iQ})\rho(e^{4iQ})\nonumber\\
&=&\rho(e^{4iP}e^{4iQ}e^{4iR}).
\end{eqnarray}
Let $e^{4iM}\equiv e^{4iP}e^{4iQ}e^{4iR}$ which can be numerically computed efficiently. Then we have $\hat{T}=\exp\left(M_{ij} i \hat{c}_i\hat{c}_j\right)$ with $M\in \mathfrak{so}(2N)$. Using a (complex) orthogonal transformation, we can bring $M$ to a block diagonal form 
\begin{equation}
O^TM O=\mathbf{diag}\left\{
\begin{pmatrix}
0 & -\epsilon_1\\
\epsilon_1 & 0
\end{pmatrix},
\begin{pmatrix}
0 & -\epsilon_2\\
\epsilon_2 & 0
\end{pmatrix},
\ldots,
\begin{pmatrix}
0 & -\epsilon_N\\
\epsilon_N & 0
\end{pmatrix}
\right\}
\end{equation}
where $\{\epsilon_j\}^N_{j=1}$ are complex numbers with non-negative real part, and $\hat{T}$ factorizes into a product of mutually commuting operators. The principal eigenvalue of $\hat{T}$ is%
\begin{eqnarray}
\Lambda_{\max}= e^{2(\epsilon_1+\epsilon_2+\ldots+\epsilon_N)}.\nonumber\\
\end{eqnarray}
The vortex excitation gap $\Delta$ is defined as
\begin{equation}\label{def:vgap}
	\Delta=\max_{(\Phi_x,\Phi_y)}\ln\Lambda_{\mathrm{max}}^{(\Phi_x,\Phi_y)}-\max_{V} \ln \Lambda_\mathrm{max}^{(V)}.
\end{equation}
where $\Lambda_{\mathrm{max}}^{(\Phi_x,\Phi_y)}$ is principal eigenvalue of the vortex free sectors defined in Eq.~\eqref{eq:lamdba_max}, and the second $\mathrm{max}$ is over all vortex configurations $V$.

\subsection{Numerical results}\label{appen:numerical_vgap}
We present numerical results that show the excitation gaps of vortices, as defined by Eq.~\eqref{def:vgap}, remain finite for $L\to\infty$. 

For numerical convenience we use the lattice orientation shown in Fig.~\ref{fig:2dhc_square-sym-1}. %
This slightly changes the finite size results from an $L\times L$ system in Fig.~\ref{fig:brickwall-unitcell}, but the thermodynamic limit remains the same. [Also notice that the  generalized Lieb's theorem in App.~\ref{appen:Lieb} still holds here since we still have reflection positivity with reflection mirrors being vertical bisectors of the $z$-links.]

We limit our numerical study to the region $K_x, K_y, K_z\leq 1.0$. 
There are $2^{L^2-1}$  vortex configurations in total, and it is impractical to study all of them, so we compared a few representative ones, including configurations with a few neighboring vortices, configurations with two far separated vortices, and configurations with a periodic vortex lattice. 
Our result shows that the excitation gap $\Delta$ increases with the number of vortices, and for a fixed number of vortices, $\Delta$ typically increases with their distance. Vortex lattices always have a finite energy density, i.e. $\Delta\propto L^2$.

In short, in all the configurations we have studied, the ones with  smallest excitation gap are configurations with two neighboring vortices, with $\vec{r}_1-\vec{r}_2=(1,-1)$ or $(1,1)$, where $\vec{r}_1$ and $\vec{r}_2$ are positions of the two vortices. %
In Fig.~\ref{fig:E2v_FSE} we show the finite size scaling of the excitation gap $\Delta_{2v}$ of two neighboring vortices with $\vec{r}_1-\vec{r}_2=(1,-1)$, for the $A$-phase, the gapless $B$-phase~($\kappa=0$) and the gapped $B$-phase~($\kappa>0$). (We also studied the finite size scaling of a few other configurations with two or four vortices, and saw similar behaviors). We see that in all cases presented here, $\Delta_{2v}$ converges to a finite positive value when $L\to\infty$, verifying our claim that vortices are always gapped. 
\begin{figure}                                                                                                            
	\center{\includegraphics[width=0.5\linewidth]{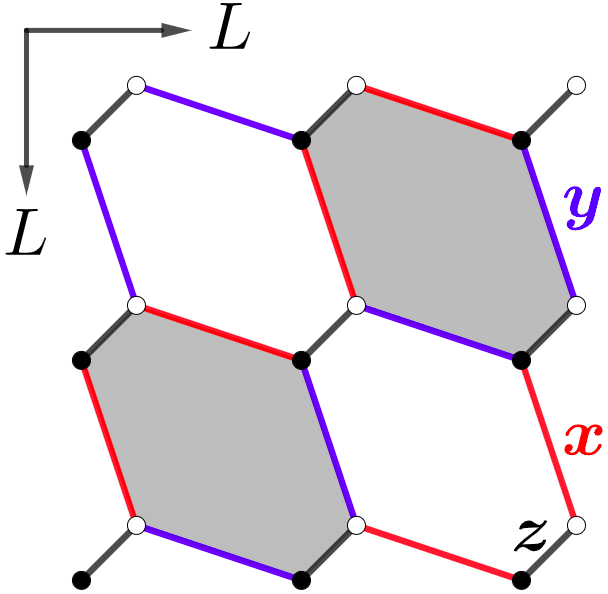}}
	\caption{\label{fig:2dhc_square-sym-1} Lattice geometry for numerical calculation of vortex excitation gaps. Both directions are periodic with length $L$. The two shaded plaquettes are the locations of the two vortices with $\vec{r}_1-\vec{r}_2=(1,-1)$ whose excitation gap is shown in Fig.~\ref{fig:E2v_FSE}. }
\end{figure}
\begin{figure}              
	\center{\includegraphics[width=0.8\linewidth]{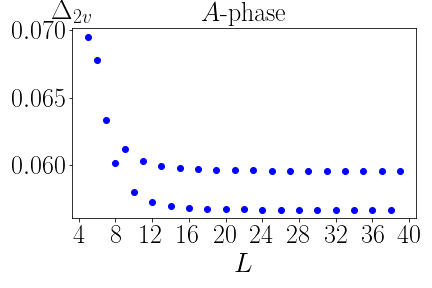}}                                                                                          
	\center{\includegraphics[width=0.8\linewidth]{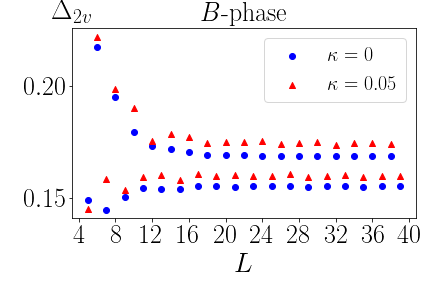}}
	\caption{\label{fig:E2v_FSE} Finite size scaling of vortex pair excitation gap $\Delta_{2v}$ as a function of system size $L$, for a pair of neighboring vortices shown in Fig.~\ref{fig:2dhc_square-sym-1}. Upper: $A_z$-phase at $K_x=K_y=0.4, K_z=1.0$, and $\kappa=0$; Lower: $B$-phases at $K_x=K_y=K_z=0.5$, with blue dots for $\kappa=0$~(where fermion spectrum is gapless) and red triangles for $\kappa=0.05$~(where fermion spectrum is gapped), respectively.}
\end{figure}

\section{Topological degeneracy of the transfer matrix $\hat{T}$ for arbitrary $(L_x,L_y)$}\label{appen:TPD}
In Sec.~\ref{sec:TPloop} we computed the topological degeneracy of $\hat{T}$ when $(L_x,L_y)$ are both even numbers. We treat the slightly more complicated case of arbitrary $(L_x,L_y)$ here. The %
results for the loop observables remain the same. 

For general $(L_x,L_y)$, the mode $(q_x, q_y)$ appears in the sector $(\Phi_x,\Phi_y)=(e^{iq_xL_x}, e^{iq_yL_y})$, for $\vec{q}=(0,0),(0,\pi),(\pi,0),(\pi,\pi)$. The value of $P_{\vec{q}}$ in the fermionic principal state is still determined by maximizing the $\tilde{T}_{0,\vec{q}}$ term in Eq.~\eqref{eq:T_0qxqy}
and we still have $P_{\pi0}=P_{\pi\pi}=-1$, $P_{00}=[K_z>K_x+K_y]$, and $P_{0\pi}=[K_z>|K_x-K_y|]$, where $[S]=+1$ if the statement $S$ is true and $[S]=-1$ otherwise.  We can rewrite the fermion parity constraint Eq.~\eqref{eq:parity_restriction2} as follows
\begin{equation}\label{eq:parity_restriction3}
\prod_{\vec{q}\equiv-\vec{q}} P_{\vec{q}}^{\delta_{\Phi_x,e^{iq_xL_x}}\delta_{\Phi_y,e^{iq_yL_y}}} 
	= \Phi_x^{L_y}(-1)^{(L_x-1)L_y},
\end{equation}
where the product is over all $\vec{q}\in\{(0,0),(0,\pi),(\pi,0),(\pi,\pi)\}$, but the exponent in $P_{\vec{q}}$ ensures that only those $\vec{q}$ belonging to the sector $(\Phi_x,\Phi_y)$ contribute. We can simplify the above equation further by the identity
\begin{equation}\label{eq:Phieiql}
	\delta_{\Phi_\alpha,e^{iq_\alpha L_\alpha}}\equiv \delta_{q_\alpha,\pi} L_\alpha+\delta_{\Phi_\alpha,1}~~~~(\mathrm{mod} ~2).
\end{equation}
Inserting Eq.~\eqref{eq:Phieiql} and the expressions of $P_{00},P_{0\pi},P_{\pi 0},P_{\pi\pi}$ given above into Eq.~\eqref{eq:parity_restriction3}, the fermion parity constraint becomes
\begin{equation}\label{eq:parity_restriction4}
	[K_z>K_x+K_y]^{\delta_{\Phi_x 1}\delta_{\Phi_y 1}}=	[K_z>|K_x-K_y|]^{\delta_{\Phi_x 1}(L_y+\delta_{\Phi_y 1})}.
\end{equation}

We can now determine the degeneracy of different phases using Eq.~\eqref{eq:parity_restriction4}:\\
$A_z$: $[K_z>K_x+K_y]=[K_z>|K_x-K_y|]=+1$, always has 4-fold degeneracy.\\
$A_x,A_y$: $[K_z>K_x+K_y]=[K_z>|K_x-K_y|]=-1$. Has 4-fold degeneracy if $L_y$ is even, and 2-fold degeneracy if $L_y$ is odd with $(++),(+-)$ being parity-incompatible.\\
$B$: $[K_z>K_x+K_y]=-1,[K_z>|K_x-K_y|]=+1$, always has 3-fold degeneracy with $(++)$ being parity-incompatible.\\
The calculation for large loop observables remain the same as  done in the main text, leading to the same results independent of $(L_x,L_y)$. 

\bibliography{3DIsing-Kitaev}
\end{document}